\documentclass[a4paper,11pt]{article}
\pdfoutput=1
\usepackage[margin=1in]{geometry}
\usepackage{amssymb}
\usepackage{amsmath}
\usepackage{times}
\usepackage[T1]{fontenc}

\vfuzz \hfuzz
\newcommand\buchi{B\"{u}chi}
\newcommand\infi{\mathsf{infin}}
\newcommand\pri{\mathsf{pri}}
\newcommand\host{\mathsf{host}}

\newcommand\os{\mathsf{os}}
\newcommand\pred{\mathsf{pred}}
\newcommand\rename{\mathsf{rename}}
\newcommand\T{\mathcal T}

\newtheorem{theorem}{Theorem}[section]
\newtheorem{lemma}[theorem]{Lemma} 
 
\newtheorem{corollary}[theorem]{Corollary} 
\newtheorem{definition}[theorem]{Definition}

\newcommand\qed{\hfill $\Box$}

\newenvironment{proof}{\noindent\textbf{Proof. }}{\nopagebreak
 \qed\medskip}

\newenvironment{proofidea}{\noindent\textit{Proof idea. }}{\nopagebreak
  \medskip}

\newcommand\im{\ensuremath{{i_{\min}}}}

\newcommand\nc{\ensuremath{{[c]}}}

\newcommand\opt{\ensuremath{{\mathsf{opt}}}}
\newcommand\p{\ensuremath{{v}}}
\newcommand\reach{\ensuremath{{\mathsf{reach}}}}

\begin{document}

\title{Determinising Parity Automata}
\author{Sven Schewe and Thomas Varghese \\ Department of Computer Science, University of Liverpool}
\date{}

\pagestyle{plain}

\bibliographystyle{plain}

\maketitle

\begin{abstract}
Parity word automata and their determinisation play an important role in automata and game theory.
We discuss a determinisation procedure for nondeterministic parity automata through deterministic Rabin to deterministic parity automata.
We prove that the intermediate determinisation to Rabin automata is optimal.
We show that the resulting determinisation to parity automata is optimal up to a small constant.
Moreover, the lower bound refers to the more liberal Streett acceptance.
We thus show that determinisation to Streett would not lead to better bounds than determinisation to parity.  
As a side-result, this optimality extends to the determinisation of B\"uchi automata.
\end{abstract}

\section{Introduction}

Church's realisability problem \cite{Church/63/Logic} has motivated the development of the beautiful theories of finite games of infinite duration \cite{Buchi/62/Automata,Buchi+Landweber/69/MSO}, and finite automata over infinite structures \cite{Rabin/69/Automata}.
These two fields have often inspired and influenced each other.

The quest for optimal complementation \cite{Vardi/07/Saga,Yan/08/lowerComplexity,Schewe/09/complementation} and determinisation \cite{Rabin/69/Automata,Safra/88/Safra,Piterman/07/Parity,Schewe/09/determinise,Colcombet+Zdanowski/09/Buchi} of nondeterministic automata has been long and fruitful.
The quest for optimal \buchi\ complementation techniques seems to have been settled with matching upper \cite{Schewe/09/complementation} and lower \cite{Yan/08/lowerComplexity} bounds.
A similar observation might, on first glance, be made for \buchi\ determinisation, as matching upper \cite{Schewe/09/determinise} and lower \cite{Colcombet+Zdanowski/09/Buchi} bounds were established shortly after those for complementation.
However, while these bounds are tight to the state, they refer to deterministic Rabin automata only, with an exponential number of Rabin pairs in the states of the initial \buchi\ automaton.

Choosing Rabin automata as targets is not the only natural choice.
The dual acceptance condition, suggested by Streett \cite{Streett/82/Streett}, would be a similarly natural goal, and determinising to parity automata seems to be an even more attractive target, as emptiness games for parity automata \cite{Schewe/07/parity,Paterson+Zwick/08/Parity} have a lower computational complexity compared to emptiness games for Streett or Rabin automata \cite{Piterman+Pnueli/06/Rabin}.
For parity and Streett automata, however, no similarly tight result is known. Indeed, the best known algorithm \cite{Piterman/07/Parity} provides an $O(n!^2)$ bound on the states \cite{Schewe/09/determinise} (for state-based acceptance; the bound can be improved to $O(n!(n-1)!)$ when transition based acceptance is used) of a deterministic parity automaton obtained from a nondeterministic \buchi\ automaton with $n$ states, as compared to the approximately $(1.65n)^n$ states of the smallest deterministic Rabin automaton \cite{Schewe/09/determinise,Colcombet+Zdanowski/09/Buchi}.

Another argument for using parity or Streett conditions is that determinisation constructions are often nested.
E.g., in distributed synthesis \cite{Pnueli+Rosner/90/Distributed,Kupferman+Vardi/01/Synthesizing,Finkbeiner+Schewe/05/Distributed},
several co-determinisation (determinisation of the complement language) steps are used.
Using Rabin automata as a target in one step, one has to use a determinisation technique for Streett automata in the next.
Streett determinisation, however, is significantly more involved and expensive \cite{Safra/92/Streett,Piterman/07/Parity}.

In this paper, we introduce determinisation procedures for nondeterministic parity automata to deterministic Rabin and parity automata.
Using an algorithmic representation that extends the determinisation procedures from \cite{Schewe/09/determinise}, we show that the number of states used in the determinisation of nondeterministic \buchi\ automata cannot be reduced by a single state, while we establish the tightness of our parity determinisation procedure to below a constant factor of $1.5$, even if we allow for Streett acceptance.
This also shows that determinising parity automata to Rabin automata leads to a smaller blow-up than the determinisation to parity or Streett.
As a special case, this holds in particular for B\"uchi automata.

\paragraph*{Transition-based acceptance.}
We use a transition based acceptance mechanism for various reasons.
Transition-based acceptance mechanisms have proven to be a more natural target of automata transformations.
Indeed, all determinisation procedures quoted above have a natural representation with an acceptance condition on transitions, and their translation to state-based acceptance is by multiplying the acceptance from the last transition to the statespace.
A similar observation can be made for other automata transformations, like the removal of $\varepsilon$-transitions from translations of $\mu$-calculi \cite{Wilke/01/Alternating,Schewe+Finkbeiner/06/ATM} and the treatment of asynchronous systems \cite{Schewe+Finkbeiner/06/Asynchronous}, where the statespace grows by multiplication with the acceptance information (e.g., maximal priority on a finite sequence of transitions), while it can only shrink in case of transition based acceptance.
Similarly, tools like \textsc{SPOT} \cite{DuretLutz11} offer more concise automata with transition-based acceptance mechanism as a translation from LTL.
Using state-based acceptance in the automaton that we want to determinise would also complicate the presentation.
But first and foremost, using transition based acceptance provides cleaner results.

\paragraph*{Related work.}
Besides the work on complementing \cite{Vardi/07/Saga,Yan/08/lowerComplexity,Schewe/09/complementation} and determinising \cite{Rabin/69/Automata,Safra/88/Safra,Piterman/07/Parity,Schewe/09/determinise,Colcombet+Zdanowski/09/Buchi} \buchi\ automata, tight bounds have been obtained for generalised \buchi\ automata \cite{Schewe+Varghese/12/GBA}, and specialised algorithms for complementing \cite{Cai+Zhang/11/Streett} and determinising Streett \cite{Safra/92/Streett,Piterman/07/Parity} automata have been studied.

The NP-completeness of minimising deterministic \buchi\ or parity automata \cite{Schewe/10/minimise} suggests that it would be infeasible to look for polynomially bigger automata and to minimise them subsequently, whereas minimising the number of priorities of a deterministic parity automaton is cheap and simple~\cite{Carton+Maceiras/99/ParityIndex}.

The construction of deterministic Co\buchi\ automata with a one-sided error, which is correct for Co\buchi\ recognisable languages \cite{Boker+Kupferman/09/cobuchi}, and decision procedures that use emptiness equivalent B\"uchi~\cite{Kupferman+Vardi/05/Safraless,KPV/06/Safraless}  or  safety \cite{Finkbeiner+Schewe/12/bounded} automata instead of language equivalent automata 
have also been studied.

\section{Preliminaries}

We denote the set of non-negative integers by $\omega$, i.e. $\omega = \{0,1,2,3,...\}$.
For a finite alphabet $\Sigma$, we use $\Sigma^*$ to denote the set of finite sequences over $\Sigma$, $\Sigma^+$ to denote the set of finite non-empty sequences over $\Sigma$, and $\Sigma^{\omega}$ to denote the set of infinite sequences over $\Sigma$.
An infinite \emph{word} $\alpha: \omega \rightarrow \Sigma$ is an infinite sequence of letters $\alpha_0 \alpha_1 \alpha_2 \cdots$ from $\Sigma$.
We use $[k]$ to represent $\{1,2,\ldots,k\}$.

$\omega$-automata are finite automata that are interpreted over infinite words and recognise $\omega$-regular languages $L\subseteq \Sigma^{\omega}$.
Nondeterministic $\omega$-automata are quintuples $\mathcal{N}=(Q,\Sigma,I,T,\mathcal{F})$, where $Q$ is a finite set of states with a non-empty subset $I\subseteq Q$ of initial states, $\Sigma$ is a finite alphabet, $T:Q \times \Sigma \times Q$ is a transition relation that maps states and input letters to sets of successor states, and $\mathcal{F}$ is an acceptance condition.
In this paper, we consider Rabin, Streett, parity, and \buchi\ acceptance.

A \emph{run} $\rho$ of a nondeterministic $\omega$-automaton $\mathcal{N}$ on an input word $\alpha$ is an infinite sequence $\rho: \omega\rightarrow Q$ of states of $\mathcal{N}$, also denoted $\rho = q_0 q_1 q_2\cdots \in Q^{\omega}$, such that the first symbol of $\rho$ is an initial state $q_0\in I$ and, for all $i\in\omega$, $(q_{i},\alpha_{i},q_{i+1})\in T$ is a valid transition.
For a run $\rho$ on a word $\alpha$, we denote with $\overline{\rho}: i \mapsto \big(\rho(i),\alpha(i),\rho(i+1)\big)$ the transitions of $\rho$.
Let $\infi(\rho) = \{q\in Q \mid \forall i \in \omega \; \exists j>i\mbox{ such that } \rho(j)=q\}$ denote the set of all states that occur infinitely often during the run $\rho$.
Likewise, let $\infi(\overline{\rho})=\{t\in T \mid \forall i \in \omega \; \exists j>i\mbox{ such that } \overline{\rho}(j)=t\}$ denote the set of all transitions that are taken infinitely many times in $\rho$.

In this paper, we use acceptance conditions over transitions. Acceptance mechanisms over states can be defined accordingly.
\emph{Rabin} automata are $\omega$-automata, whose acceptance is defined by a family of pairs $\{(A_i,R_i)\mid i \in J\}$, with $A_i,R_i \subseteq T$, of accepting and rejecting transitions for all indices $i$ of some index set $J$.
A run $\rho$ of a Rabin automaton is \emph{accepting} if there is an index $i\in J$, such that infinitely many accepting transitions $t \in A_i$, but only finitely many rejecting transitions $t \in R_j$ occur in $\overline{\rho}$.
That is, if there is an $i\in J$ such that $\infi(\overline{\rho}) \cap A_i \neq \emptyset = \infi(\overline{\rho}) \cap R_i$.
\emph{Streett} automata are $\omega$-automata, whose acceptance is defined by a family of pairs $\{(G_i,B_i)\mid i \in J\}$, with $G_i,B_i \subseteq T$, of good and bad transitions for all indices $i$ of some index set $J$.
A run $\rho$ of a Streett automaton is \emph{accepting} if, for all indices $i\in J$, some good transition $t \in G_i$ or no bad transition $t \in B_j$ occur infinitely often in $\overline{\rho}$.
That is, if, for all $i\in J$, $\infi(\overline{\rho}) \cap G_i \neq \emptyset$ or $\infi(\overline{\rho}) \cap B_i = \emptyset$ holds.

\emph{Parity} automata are $\omega$-automata, whose acceptance is defined by a priority function $\pri:T \rightarrow [c]$ for some $c \in \mathbb N$.
A run $\rho$ of a parity automaton is \emph{accepting} if $\limsup_{n\rightarrow\infty} \pri\big(\overline{\rho}(n)\big)$ is even, that is, if the highest priority that occurs infinitely often is even.
Parity automata can be viewed as special Rabin, or as special Streett automata.
In older works, the parity condition was referred to as Rabin chain condition---because one can represent them by choosing $A_i$ as the set of states with priority $\leq 2i$ and $R_i$ as the sets of states with priorities $\leq 2i -1$, resulting in a chain $A_i \subseteq R_i \subseteq A_{i+1} \subseteq \ldots$---or a Streett chain condition---where $G_i$ is the set of states with priority $\geq 2i$, and $B_i$ is the set of states with priority $\geq 2i-1$.

One-pair Rabin automata $\mathcal R_1 =\big(Q,\Sigma,I,T,(A,R)\big)$, which are Rabin automata with a singleton index set, such that we directly refer to the only pair $(A,R)$, and \emph{B\"uchi} automata, which can be viewed as one-pair Rabin automata with an empty set of rejecting states $R = \emptyset$, are of special technical interest in this paper.

For all types of automata, a word $\alpha$ is accepted by an automaton $\mathcal A$ iff it has an accepting run, and its language $\mathcal{L}(\mathcal A)$ is the set of words it accepts.

We call an automaton $(Q,\Sigma,I,T,\mathcal{F})$ \emph{deterministic} if $I$ is singleton and $T$ contains at most one target node for all pairs of states and input letters, that is, if $(q,\alpha,r),(q,\alpha,s) \in T$ implies $r = s$.
Deterministic automata are denoted $(Q,\Sigma,q_0,\delta,\mathcal{F})$, where $q_0$ is the only initial state and $\delta$ is the partial function with $\delta: (q,\alpha) \mapsto r \Leftrightarrow (q,\alpha,r)\in T$.

As nondeterministic automata can block, we also allow them to accept immediately.
Technically, one can use a state $\top$ which every automaton has.
From $\top$, all transitions go back to $\top$, and sequences that contain one (and thus almost only) $\top$ states are accepting.
This state is not counted to the statespace $Q$.
If we want to include it, we explicitly write $Q^\top$.

\section{Determinisation}

A nondeterministic parity automaton $\mathcal{P}$ is a quintuple $\big(P,\Sigma,I,T,\pri:T \rightarrow [c] \big)$. This NPA has $|P|=n$ states and $c$ priorities (or colours) on the transitions.

We will tackle the determinisation of parity automata in three steps.
Firstly, we will recall history trees, the data structure for determinising \buchi\ automata.
Secondly, we will describe an adjustment of the data structure and a determinisation procedure from \buchi\ automata to one-pair Rabin automata.
Finally, we will show that this data structure can be nested for the determinisation of parity automata.

In \cite{Schewe/09/determinise,Schewe+Varghese/12/GBA}, we use ordered labelled trees to depict the states of the deterministic automaton. 
These ordered labelled trees are called \emph{history trees} in \cite{Schewe/09/determinise,Schewe+Varghese/12/GBA}. 

A \emph{history tree} is an ordered labelled tree $(\T,l)$, where $\T$ is a finite, prefix closed subset of finite sequences of natural numbers $\omega$.
Every element $v\in \T$ is called a \emph{node}.
Prefix closedness implies that, if a node $v=n_{1}\ldots n_{j}n_{j+1} \in \T$ is in $\T$, then $v'=n_{1}\ldots n_{j}$ is also in $\T$.
We call $v'$ the predecessor of $v$, denoted $\pred(v)$.
The empty sequence $\epsilon \in \T$ is called the \emph{root} of the ordered tree $\T$.
Obviously, $\epsilon$ has no predecessor. 

We further require $\T$ to be \emph{order closed} with respect to siblings:
if a node $v=n_{1}\ldots n_{j}$ is in $\T$, then $v'=n_{1}\ldots n_{j-1}i$ is also in $\T$ for all $i\in \omega$ with $i<n_j$.
In this case, we call $v'$ an \emph{older sibling} of $v$ (and $v$ a \emph{younger sibling} of $v'$).
We denote the set of older siblings of $v$ by $\os(v)$.

A history tree is a tree labelled with sets of automata states.
That is, $l:\T\rightarrow 2^{Q} \smallsetminus \{\emptyset\}$ is a labelling function, which maps nodes of $\T$ to non-empty sets of automata states.
For \buchi\ automata, the labelling is subject to the following criteria.
\begin{enumerate}
 \item The label of each node is a subset of the label of its predecessor:
 \hfill
  $l(v)\subseteq l(\pred(v))$ holds for all $\varepsilon \neq v \in \T$.
  
 \item The intersection of the labels of two siblings is disjoint:
  
  $\forall v,v' {\in} \T.\ v {\neq} v' \wedge \pred(v) {=} \pred(v') \Rightarrow l(v) {\cap} l(v') = \emptyset$.
  \item The union of the labels of all siblings is \emph{strictly} contained in the label of their predecessor:
  \hfill
  $\forall v \in \T\ \exists q \in l(v)\ \forall v' \in \T.\ v=\pred(v') \Rightarrow q \notin l(v')$.
\end{enumerate}

\subsection{Determinising one-pair Rabin automata}
\label{Rabindet}

For one-pair Rabin automata, it suffices to adjust this data structure slightly.
A tree is called a \emph{root history tree} (RHT) if it satisfies (1) and (2) from the definition of history trees, and a relaxed version of (3) that allows for non-strict containment of the label of the root,
$\forall v \in \T \mathbf{\mathbf \smallsetminus \{\varepsilon\}}\ \exists q \in l(v)\ \forall v' \in \T.\ v=\pred(v') \Rightarrow q \notin l(v')$, and the label of the root $\varepsilon$ \emph{equals} the union of its children's labels,
$l(\varepsilon) = \bigcup\{l(v) \mid v \in \T \cap \omega\}$. 

Let $\mathcal{R}_1 = (Q,\Sigma,I,T,(A,R))$ be a nondeterministic one-pair Rabin automaton with $|Q| = n$ states. We first construct a language equivalent deterministic Rabin automaton $\mathcal D_1 =(D,\Sigma,d_{0},\Delta,\{(A_i,R_i) \mid i \in J\})$ where,
\begin{itemize}
 \item $D$ is the set of RHTs over $Q$,
 \item $d_{0}$ is the history tree $(\{\varepsilon,0\},\ l:\varepsilon \mapsto I,\ l:0 \mapsto I)$,
 \item $J$ is the set of nodes $\neq \varepsilon$ that occur in some RHT of size $n+1$ (due to the definition of RHTs, an RHT can contain at most $n+1$ nodes), and
 \item for every tree $d\in D$ and letter $\sigma\in\Sigma$, the transition $d'=\Delta(d,\sigma)$ is the result of the sequence of the transition mechanism described below.

 The index set is the set of nodes, and, for each index, the accepting and rejecting sets refer to this node.
\end{itemize}

\subsubsection*{\textbf{Transition mechanism for determinising one-pair Rabin Automata}} 

We determine $\Delta{:}\big((\T,l),\sigma\big) \mapsto (\T',l')$~as~follows:%
\hspace*{-2mm}
\label{TransMech}
\begin{enumerate}
\item \emph{Update of node labels (subset constructions).}
The root of a history tree $d$ collects the momentarily reachable states $Q_r \subseteq Q$ of the automaton $\mathcal{R}_1$.
In the first step of the construction, we update the label of the root to the set of reachable states upon reading a letter $\sigma \in \Sigma$, using the classical subset construction. We update the label of every other node of the RHT $d$ to reflect the successors reachable through accepting or neutral transitions.

For $\varepsilon$, we update $l$ to the function $l_1$ by assigning $l_1: \varepsilon \mapsto \{q' \in Q \mid \exists q \in l(\varepsilon).\ (q,\sigma,q') \in T\}$, and
for all $\varepsilon \neq v\in \T$, we update $l$ to the function $l_1$ by assigning $l_1: v \mapsto \{q' \in Q \mid \exists q \in l(v).\ (q,\sigma,q') \in T \smallsetminus R\}$.

\item \emph{Splitting of run threads / spawning new children.}
In this step, we spawn new children for every node in the RHT.
For nodes other than the root $\varepsilon$, we spawn a child labelled with the set of states reachable through accepting transitions; for the root $\varepsilon$, we spawn a child labelled like the root.

Thus, for every node $\varepsilon \neq v\in d$ with $c$ children, we spawn a new child $vc$ and expand $l_1$ to $vc$ by assigning  $l_1: vc \mapsto \{q \in Q \mid \exists q' \in l(v).\ (q',\sigma,q) \in A\}$.
If $\varepsilon$ has $c$ children, we spawn a new child $c$ of the root $\varepsilon$ and expand $l_1$ to $c$ by assigning  $l_1: c \mapsto l_1(\varepsilon)$.
We use $\mathcal \T_n$ to denote the extended tree that includes the new children.
 
\item \emph{Removing states from labels -- horizontal pruning.}
We obtain a function $l_2$ from $l_1$ by removing, for every node $v$ with label $l(v)=Q'$ and all states $q\in Q'$, $q$ from the labels of all younger siblings of $v$ and all of their descendants.
 
\item \emph{Identifying breakpoints -- vertical pruning.}
We denote with $\T_e\subseteq \T_n$ the set of all nodes $v\neq \varepsilon$ whose label $l_2(v)$ is now equal to the union of the labels of its children.
We obtain $\T_v$ from $\T_n$ by removing all descendants of nodes in $\T_e$, and restrict the domain of $l_2$ accordingly.

Nodes in $\T_v \cap \T_e$ represent the breakpoints reached during the infinite run $\rho$ and are called \emph{accepting},
that is, the transition of $\mathcal D_1$ will be in $A_v$ for exactly the $v \in \T_v \cap \T_e$.
Note that the root cannot be accepting.

\item \emph{Removing nodes with empty label.} We denote with $\T_r = \{v \in \T_v \mid l_2(v) \neq \emptyset\}$ the subtree of $\T_v$ that consists of the nodes with non-empty label and restrict the domain of $l_2$ accordingly.

\item \emph{Reordering.}
To repair the orderedness, we call $\|v\| = |\os(v) \cap \T_r|$ the number of (still existing) older siblings of $v$, and map
$v= n_1\ldots n_j$ to 
$v' = \|n_1\|\ \|n_1n_2\|\ \|n_1n_2n_3\|\ldots \|v\|$, denoted $\rename(v)$.

For $\T'=\rename(\T_r)$, we update a pair $(\T_r,l_2)$ from Step 5 to
$d'=\big(\T', l' \big)$ with $l': \rename(v) \mapsto l_2(v)$.

We call a node $v\in \T'\cap \T$ \emph{stable} if $v=\rename(v)$, and we call all nodes in $J$ \emph{rejecting} if they are not stable.
That is, the transition will be in $R_v$ exactly for those $v \in J$, such that $v$ is not a stable node in $\T \cap \T'$.
\end{enumerate}

Note that this construction is a generalisation of the same construction for \buchi\ automata: if $R = \emptyset$, then the label of $0$ is always the label of $\varepsilon$ in this construction, and the node $1$ is not part of any reachable RHT.
(We would merely write $0$ in front of every node of a history tree.)

The correctness proof of this construction follows the same lines as the correctness proof of the \buchi\ construction.

\begin{lemma}
\label{lemma:determ1}
[$L(\mathcal{R}_1) \subseteq L(\mathcal D_1)$] Given that there is an accepting run of $\mathcal{R}_1$ on an $\omega$-word $\alpha$, there is a node $v \in J$ that is eventually always stable and always eventually accepting in the run of $\mathcal D_1$ on $\alpha$.
\end{lemma}

\paragraph*{Notation.}
For a state $q$ of $\mathcal{R}_1$ and an RHT $d=(\T,l)$, we call a node $v$ the \emph{host node of $q$}, denoted $\host(q,d)$, if $q \in l(v)$, but not in $l(vc)$ for any child $vc$ of $v$.

\begin{proofidea} This is the same as for \buchi\ determinisation \cite{Schewe/09/determinise}:
the state of each accepting run is eventually `trapped' in the same node of the RHT, and this node must be accepting infinitely often.
Let $d_0,d_1 \ldots$ be the run of $\mathcal D_1$ on $\alpha$ and $q_0, q_1, \ldots$ an accepting run of $\mathcal{R}_1$ on $\alpha$.
Then we can define a sequence $v_0,v_1, \ldots$ with $v_i = \host(q_i,d_i)$, and there must be a longest eventually stable prefix $v$ in this sequence.

An inductive argument can then be exploited to show that, once this prefix $v$ is henceforth stable, the index $v$ cannot be rejecting.
The assumption that there is a point in time where $v$ is stable but never again accepting can lead to a contradiction.
Once the transition $(q_i,\alpha(i),q_{i+1})$ is accepting, $q_{i+1} \in l_{i+1}(vc)$ for some $c \in \omega$ and for $d_{i+1}=(\T_{i+1},l_{i+1})$.
As $v$ is never again accepting \emph{or} rejecting, we can show for all $j>i$ that, if $q_j \in l_j(vc_j)$, then $q_{j+1} \in l_{j+1}(vc_{j+1})$ for some $c_{j+1} \leq c_j$.
This monotonicity leads to a contradiction with the assumption that $v$ is the \emph{longest} stable prefix.
\end{proofidea}

\begin{proof}
We fix an accepting run $\rho=q_{0}q_{1} \ldots$ of $\mathcal{R}_1$ on an input word $\alpha$, and let $\rho_{\mathcal D_{1}}=d_{0}d_{1}\ldots$ be the run of $\mathcal D_{1}$ on $\alpha$.
We then define the related sequence of host nodes $\vartheta=v_{0}v_{1}v_{2}\ldots=\host(q_0,d_0)\host(q_1,d_1)\host(q_2,d_2)\ldots$.

Let $s=\liminf_{n\rightarrow\infty} |v_{n}|$ be the shortest length of these nodes that occurs infinitely often.
Note that the root cannot be the host node of any state, as it is always labelled by the union of the labels of its children.

We follow the run and argue that the initial sequence of length $s$ of the nodes in $\vartheta$ eventually stabilises.
Let $i_{0} < i_{1} < i_{2} <\ldots$ be an infinite ascending chain of indices such that
\begin{enumerate}
 \item $(q_j,\alpha(j),q_{j+1})\in T \smallsetminus R$ is a neutral or accepting transition for any $j\geq i_0$,
 \item the length $|v_{j}|\geq s$ of the j-th node is not smaller than $s$ for all $j \geq i_{0}$, and
 \item the length $|v_{j}| = s$ is equal to $s$ for all indices $j \in \{i_{0},i_{1},i_{2},\ldots\}$ in this chain.
\end{enumerate}

This implies that $v_{i_{0}}, v_{i_{1}}, v_{i_{2}}, \ldots$ is a descending chain when the single nodes $v_{i}$ are compared by lexicographic order.
As the domain is finite, almost all elements of the descending chain are equal, say $v_{i}:=\pi$.
In particular, $\pi \in J$ is eventually always stable.

Let us assume for contradiction that this stable prefix $\pi$ is accepting only finitely many times.
We choose an index $i$ from the chain $i_0 < i_1 < i_2 <\ldots$ such that
\begin{enumerate}
 \item $\pi$ is stable for all $j\geq i$ and
 \item $\pi$ is not accepting for any $j \geq i$.
\end{enumerate}
Note that $\pi$ is the host of $q_i$ for $d_i$, and $q_j \in l_j(\pi)$ holds for all $j\geq i$.

As $\rho$ is accepting, there is a smallest index $j>i$ such that $(q_{j-1},\alpha(j-1),q_j)\in A$.
Now, as $\pi$ is stable but not accepting for all $k\geq i$ (and hence for all $k \geq j$), $q_k$ must henceforth be in the label of a child of $\pi$ in $d_k$, which contradicts the assumption that infinitely many nodes in $\vartheta$ have length $s = |\pi|$.

Thus, $\pi$ is eventually always stable and always eventually accepting.
\end{proof}

\begin{lemma}
\label{lemma:determ2}
[$L(\mathcal D_1) \subseteq L(\mathcal{R}_1)$]
Given that there is a node $v\in J$, which is eventually always stable and always eventually accepting for an $\omega$-word $\alpha$, then there is an accepting run of $\mathcal{R}_1$ on $\alpha$. 
\end{lemma}

\paragraph*{Notation.} For an $\omega$-word $\alpha$ and $j\geq i$, we denote with $\alpha[i,j[$ the word 
$\alpha(i)\alpha(i+1)\alpha(i+2)\ldots\alpha(j-1)$.

We denote with $Q_1 \rightarrow^\alpha Q_2$ for a finite word $\alpha=\alpha_1\ldots\alpha_{j-1}$ that there is, for all $q_j\in Q_2$ a sequence
$q_1\ldots q_j$ with $q_1 \in Q_1$ and $(q_i,\alpha_i,q_{i+1})\in T$ for all $1\leq i <j$.
If, for all $q_j \in Q_2$, there is such a sequence that contains a transition in $A$ but no transition in $R$, we write $Q_1 \Rightarrow^\alpha Q_2$.  

\begin{proofidea} For the run $d_0 d_1 d_2 \ldots$ of $\mathcal D_1$ on $\alpha$, we fix an ascending chain $1 < i_0 < i_1 < i_2 \ldots$ of indices, such that $v$ is not rejecting in any transition $(d_{j-1},\alpha(j-1),d_j)$ for $j \geq i_0$ and such that $(d_{i_j-1},\alpha(i_j-1),d_{i_j}) \in A_v$ for all $j \geq 0$.
The \emph{proof idea} is the usual way of building a tree of initial sequences of runs: we build a tree of initial sequences of runs of $\mathcal{R}_1$ that contains a sequence
$q_0q_1q_2\ldots q_{i_j}$ for any $j\in \omega$ iff
\begin{itemize}
 \item $(q_i,\alpha(i),q_{i+1}) \in T$ is a transition of $\mathcal{R}_1$ for all $i < i_j$,
 \item $(q_i,\alpha(i),q_{i+1}) \notin R$ is not rejecting for all $i \geq i_0 - 1$, and
 \item for all $k<j$ there is an $i\in [i_k,i_{k+1}[$ such that $(q_i,\alpha(i),q_{i+1}) \in A$ is an accepting transition.
\end{itemize}
This infinite tree has an infinite branch by K\"onig's Lemma.
By construction, this branch is an accepting run of $\mathcal{R}_1$ on $\alpha$.
\end{proofidea}

\begin{proof}
Let $\alpha \in L(\mathcal D_1)$.
Then there is a $v$ that is eventually always stable and always eventually accepting in the run $\rho_{\mathcal D_{1}}$ of $\mathcal D_{1}$ on $\alpha$.
We pick such a $v$.

Let $1 < i_{0}<i_{1}<i_{2}<\ldots$ be an infinite ascending chain of indices such that
\begin{itemize}
 \item $v$ is stable for all transitions $(d_{j-1},\alpha(j-1),d_{j})$ with $j\geq i_0$, and
 \item the chain $ i_{0}<i_{1}<i_{2}<\ldots$ contains exactly those indices $i\geq i_0$ such that $(d_{i-1},\alpha(i-1),d_{i})$ is accepting.
\end{itemize}

Let $d_i=(\T_i,l_i)$ for all $i \in \omega$.
By construction, we have
\begin{itemize}
 \item $I \rightarrow^{\alpha[0,i_0[} l_{i_0}(v)$, and
 \item $l_{i_j}(v) \Rightarrow^{\alpha[i_j,i_{j+1}[} l_{i_{j+1}}(v)$.
\end{itemize}

Using this observation, we can build a tree of initial sequences of runs as follows:
we build a tree of initial sequences of runs of $\mathcal{R}_1$ that contains a sequence
$q_0q_1q_2\ldots q_{i_j}$ for any $j\in \omega$ iff
\begin{itemize}
 \item $(q_i,\alpha(i),q_{i+1}) \in T$ is a transition of $\mathcal{R}_1$ for all $i < i_j$,
 \item $(q_i,\alpha(i),q_{i+1}) \notin R$ is not rejecting for all $i \geq i_0 - 1$, and
 \item for all $k<j$ there is an $i\in [i_k,i_{k+1}[$ such that $(q_i,\alpha(i),q_{i+1}) \in A$ is an accepting transition.
\end{itemize}

By construction, this tree has the following properties:
\begin{itemize}
 \item it is infinite,
 \item it is finitely branching,
 \item no branch contains more than $i_0$ rejecting transitions, and,
 \item for all $j \in \omega$, a branch of length $>i_j$ contains at least $j$ accepting transitions.
\end{itemize}

Exploiting K\"onig's lemma, the first two properties provide us with an infinite path, which is a run of $\mathcal{R}_1$ on $\alpha$.
The last two properties then imply that this run is accepting. 
$\alpha$ is therefore in the language of $\mathcal{R}_1$.
\end{proof}

\begin{corollary}
\label{Rabineq}
$L(\mathcal{R}_1) = L(\mathcal D_1)$.
\end{corollary} 

\subsubsection*{Estimation of Root History Trees}
Let $\#\mathsf{ht}(n)$ and $\#\mathsf{rht}(n)$ be the number of history trees and RHTs, respectively, over sets with $n$ states.
First, $\#\mathsf{rht}(n) \geq \#\mathsf{ht}(n)$ holds, because the sub-tree rooted in $0$ of an RHT is a history tree.
Second, $\#\mathsf{ht}(n+1) \geq \#\mathsf{rht}(n)$, because adding the additional state to $l(\varepsilon)$ turns an RHT into a history tree.
With an estimation similar to that of history trees \cite{Schewe/09/determinise}, we get:
\begin{theorem}
\label{theo:rht}
$\inf\big\{c \mid \#\mathsf{rht}(n) \in O\big((cn)^n\big)\big\} = \inf\big\{c \mid \#\mathsf{ht}(n) \in O\big((cn)^n\big)\big\} \approx 1.65$.
\end{theorem}
In \cite{Schewe/09/determinise} it was shown that $\#\mathsf{ht}(n)$ grows at a speed, such that $\inf\big\{c \mid \#\mathsf{ht}(n) \in O\big((cn)^n\big)\big\} \approx 1.65$.
We argue that $\#\mathsf{rht}(n)$ does not only grow in the same speed, it even holds that there is only a small constant factor between $\#\mathsf{ht}(n)$ and $\#\mathsf{rht}(n)$.

First, there is obviously a bijection between RHTs over $Q$ and the subset of history trees over $Q \cup \{q_d\}$, where $q_d \notin Q$ is a fresh dummy state, and $q_d$ is the only state that is hosted by the root.
We estimate this size by the number of history trees, where $q_d$ is hosted by the root $\varepsilon$ of the history tree.

To keep the estimation simple, it is easy to see that the share of history trees with $< \frac{1}{3} n$ nodes diminishes to $0$, as the number of trees with $n$ nodes grows much faster than the number of trees with $<\frac{1}{3}n$ nodes and the number of functions from $[n]$ onto $[n]$, $n!$, grows much faster than the functions from $[n]$ to $[\frac{1}{3}n]$.
So we can assume for our estimation that the tree has at least $\frac{1}{3}n$ nodes, such that the share of trees where $q_d$ is in the root is at most $<\frac{3}{n}$.

$\lim_{n \rightarrow \infty} \frac{\#\mathsf{ht}(n+1)}{n \#\mathsf{ht}(n)}$ converges to $(1+\frac{c}{n})^n = e^c$ for $c \approx 1.65$.
Thus, we get the following estimation:

$\lim_{n \rightarrow \infty} \frac{\#\mathsf{rht}(n)}{\#\mathsf{ht}(n)} \leq  \lim_{n \rightarrow \infty} 3\frac{\#\mathsf{ht}(n+1)}{n \#\mathsf{ht}(n)} 3e^c<3 e^c$.

\subsection{Determinising parity automata}
\label{paritydet}
Having outlined a determinisation construction for one-pair Rabin automata using root history trees, we proceed to define  \emph{nested history trees} (NHTs), the data structure we use for determinising parity automata.

We assume that we have a parity automaton
$\mathcal P = (Q,\Sigma,I,T,\pri:T \rightarrow [c])$,
and we select $e = 2\lfloor 0.5 c \rfloor$.

A \emph{nested history tree} is a triple $(\T,l,\lambda)$, where $\T$ is a finite, prefix closed subset of finite sequences of natural numbers and a special symbol $\mathfrak s$ (for \emph{stepchild}), $\omega \cup \{\mathfrak s\}$.
We refer to all other children $vc$, $c\in \omega$ of a node $v$ as its \emph{natural children}.
We call $l(v)$ the label of the node $v \in \T$, and $\lambda(v)$ its \emph{level}.

A node $v \neq \varepsilon$ is called a \emph{Rabin root}, iff it ends in $\mathfrak s$.
The root $\varepsilon$ is called a Rabin root iff $c>e$.
A node $v \in \T$ is called a \emph{base node} iff it is not a Rabin root and $\lambda(v) = 2$.
The set of base nodes is denoted $\mathsf{base}(\T)$.
\begin{itemize}
 \item The label $l(v)$ of each node $v \neq \varepsilon$ is a subset of the label of its predecessor:
 
  $l(v)\subseteq l(\pred(v))$ holds for all $\varepsilon \neq v \in \T$.

 \item The intersection of the labels of two siblings is disjoint:
  
  $\forall v,v' {\in} \T.\ v {\neq} v' \wedge \pred(v) {=} \pred(v') \Rightarrow l(v) {\cap} l(v') = \emptyset$.

 \item For all \emph{base nodes}, the union of the labels of all siblings is \emph{strictly} contained in the label of their predecessor:
  
  $\forall v {\in} \mathsf{base}(\T)\ \exists q {\in} l(v)\ \forall v' {\in} \T.\ v{=}\pred(v') \Rightarrow q {\notin} l(v')$.

\item A node $v\in \T$ has a stepchild iff $v$ is neither a base-node, nor a Rabin root.

\item The union of the labels of all siblings of a non-base node \emph{equals} the union of its children's labels:
\hfill  
$\forall v {\in} \T\smallsetminus \mathsf{base}(\T)$,
$l(v) = \{q \in l(v') \mid v' \in \T \mbox{ and } v =\pred(v')\}$ holds.

\item The level of the root is $\lambda(\varepsilon)=e$.

\item The level of a stepchild is 2 smaller than the level of its parent:
\hfill
for all $v\mathfrak{s} \in \T$, $\lambda(v\mathfrak{s}) = \lambda(v)-2$ holds.

\item The level of all other children equals the level of its parent:
for all $i \in \omega$ and $vi \in \T$, $\lambda(vi) = \lambda(v)$ holds.
\end{itemize}

While the definition sounds rather involved, it is (for odd $c$) a nesting of RHTs.
Indeed, for $c = 3$, we simply get the RHTs, and $\lambda$ is the constant function with domain $\{2\}$.
For odd $c >3$, removing all nodes that contain an $\mathfrak{s}$ somewhere in the sequence again resemble RHTs,
while the sub-trees rooted in a node $v\mathfrak{s}$ such that $v$ does not contain a $\mathfrak{s}$ resemble NHTs whose root has level $c-3$.

The transition mechanism from the previous subsection is adjusted accordingly.
For each level $a$ (note that levels are always even), we define three sets of transitions for the parity automaton $\mathcal P$:
the rejecting transitions $R_a = \{t \in T \mid \pri(t)>a$ and $\pri(t)$ is odd$\}$; the accepting transitions $A_a = \{t \in T \mid \pri(t) \geq a$ and $\pri(t)$ is even$\}$, and the (at least) neutral transitions, $N_a = T \smallsetminus R_a$.

\paragraph*{Construction.} Let $\mathcal{P} = \big(P,\Sigma,I,T,\{\pri:P \rightarrow [c]\big)$ be a nondeterministic parity automaton with $|P|=n$ states.

We construct a language equivalent deterministic Rabin automaton $\mathcal{DR}=(D,\Sigma,d_{0},\Delta,\{(A_i,R_i) \mid i \in J\})$ where,
\begin{itemize}
 \item $D$ is the set of NHTs over $P$ (i.e., with $l(\varepsilon) \subseteq P$) whose root has level $e$, where $e=c$ if $c$ is even, and $e=c-1$ if $c$ is odd,
 \item $d_{0}$ is the NHT we obtain by starting with $(\{\varepsilon\},\ l:\varepsilon \mapsto I,\ \lambda: \varepsilon \mapsto e)$, and performing Step 7 from the transition construction until an NHT is produced.
 \item $J$ is the set of nodes $v$ that occur in some NHT of level $e$ over $P$, and
 \item for every tree $d\in D$ and letter $\sigma\in\Sigma$, the transition $d'=\Delta(d,\sigma)$ is the result of the sequence of transformations described below.
\end{itemize}

\paragraph*{\textbf{Transition mechanism for determinising parity automata.}}
\label{TransMechParity}
Note that we do not define the update of $\lambda$, but use $\lambda$.
This can be done because the level of the root always remains $\lambda(\varepsilon)=e$;
the level $\lambda(v)$ of all nodes $v$ is therefore defined by the number of $\mathfrak{s}$ occurring in $v$.
Likewise, the property of $v$ being a base-node or a Rabin root is, for a given $c$, a property of $v$ and independent of the labelling function.

Starting from an NHT $d=(\T,l,\lambda)$, we define the transitions $\Delta: (d,\sigma) \mapsto d'$ as follows:
\begin{enumerate}
\item \emph{Update of node labels (subset constructions):}
For the root, we continue to use $l_1(\varepsilon)=\{q' \in Q \mid \exists q \in l(\varepsilon).\ (q,\sigma,q') \in T\}$.

For other nodes $v \in \T$ that are \emph{no Rabin roots}, we use
$l_1(v)=\{q' \in Q \mid \exists q \in l(v).\ (q,\sigma,q') \in N_{\lambda(v)}\}$.

For the remaining Rabin roots $v\mathfrak{s}\in \T$, we use
$l_1(v\mathfrak{s})=\{q' \in Q \mid \exists q \in l(v\mathfrak{s}).\ (q,\sigma,q') \in N_{\lambda(v)}\}$.
That is, we use the neutral transition of the higher level of the \emph{parent} of the Rabin node.
 
\item \emph{Splitting of run threads / spawning new children.}
In this step, we spawn new children for every node in the NHT.
For nodes $v \in \T$ that are no Rabin roots, we spawn a child labelled with the set of states reachable through accepting transitions.
For a Rabin root $v \in \T$, we spawn a new child labelled like the root.

Thus, for every node $v\in \T$ which is no Rabin root and has $c$ \emph{natural} children, we spawn a new child $vc$ and expand $l_1$ to $vc$ by assigning 
$l_1: vc \mapsto \{q \in Q \mid \exists q' \in l(v).\ (q',\sigma,q) \in A_{\lambda(v)}\}$.
If a Rabin root $v$ has $c$ natural children, we spawn a new child $vc$ of the Rabin root $v$ and expand $l_1$ to $vc$ by assigning  $l_1: vc \mapsto l_1(v)$.
We use $\mathcal \T_n$ to denote the extended tree that includes the new children.

\item \emph{Removing states from labels -- horizontal pruning.}
We obtain a function $l_2$ from $l_1$ by removing, for every node $v$ with label $l(v)=Q'$ and all states $q\in Q'$, $q$ from the labels of all younger siblings of $v$ and all of their descendants.

Stepchildren are always treated as the \emph{youngest} sibling, irrespective of the order of creation.
 
\item \emph{Identifying breakpoints -- vertical pruning.}
We denote with $\T_e\subseteq \T_n$ the set of all nodes $v\neq \varepsilon$ whose label $l_2(v)$ is now equal to the union of the labels of its \emph{natural} children.
We obtain $\T_v$ from $\T_n$ by removing all descendants of nodes in $\T_e$, and restrict the domain of $l_2$ accordingly.

Nodes in $\T_v \cap \T_e$ represent the breakpoints reached during the infinite run $\rho$ and are called \emph{accepting}.
That is, the transition of $\mathcal{DR}$ will be in $A_v$ for exactly the $v \in \T_v \cap \T_e$.
Note that Rabin roots cannot be accepting.

\item \emph{Removing nodes with empty label.} We denote with $\T_r = \{v \in \T_v \mid l_2(v) \neq \emptyset\}$ the subtree of $\T_v$ that consists of the nodes with non-empty label and restrict the domain of $l_2$ accordingly.

\item \emph{Reordering.}
To repair the orderedness, we call $\|v\| = |\os(v) \cap \T_r|$ the number of (still existing) older siblings of $v$, and map
$v= n_1\ldots n_j$ to 
$v' = \|n_1\|\ \|n_1n_2\|\ \|n_1n_2n_3\|\ldots \|v\|$, denoted $\rename(v)$.

For $\T_o=\rename(\T_r)$, we update a pair $(\T_r,l_2)$ from Step 5 to
$d'=\big(\T_o, l'\big)$ with $l': \rename(v) \mapsto l_2(v)$.

We call a node $v\in \T_o\cap \T$ \emph{stable} if $v=\rename(v)$, and we call all nodes in $J$ \emph{rejecting} if they are not stable.
That is, the transition will be in $R_v$ exactly for those $v \in J$, such that $v$ is not a stable node in $\T \cap \T'$.

\item \emph{Repairing nestedness.}
We initialise $\T'$ to $\T_o$ and then add recursively for
\begin{itemize}
 \item Rabin roots $v$ without children a child $v0$ to $\T'$ and expand $l'$ by assigning $l': v0 \mapsto l'(v)$, and for
 \item nodes $v$, which are neither Rabin roots nor base-nodes, without children a child $v\mathfrak{s}$ to $\T'$ and expand $l'$ by assigning $l': v\mathfrak{s} \mapsto l'(v)$
\end{itemize}
until we have constructed an NHT $d'=(\T',l',\lambda')$.
\end{enumerate}

\begin{lemma} $L(\mathcal{P}) \subseteq L(\mathcal{DR})$
\label{deterp1}
\end{lemma}

\paragraph*{Notation.}
For a state $q$ of $\mathcal{P}$, an NHT $d=(\T,l,\lambda)$ and an even number $a \leq e$, we call a node $v'$ the \emph{$a$ host node of $q$}, denoted $\host_a(q,d)$, if $q \in l(v')$, but not in $l(v'c)$ for any natural child $v'c$ of $v'$, and $\lambda(v')=a$.

Let $\rho = q_0,q_1, q_2 \ldots$ be an accepting run of $\mathcal{P}$ with even $a = \liminf_{i \rightarrow \infty}\pri\big(q_i,\alpha(i),q_{i+1}\big)$ on an $\omega$-word $\alpha$, let $d_0d_1d_2\ldots$ be the run of $\mathcal{DR}$ on $\alpha$, and let $v_i=\host_a(q_i,d_i)$ for all $i\in \omega$.

\begin{proofidea}The core idea of the proof is again that the state of each accepting run is eventually `trapped' in a maximal initial sequence $v$ of $a$-hosts, with the additional constraint that neither $v$ nor any of its ancestors are infinitely often rejecting, and the transitions of the run of $\mathcal P$ are henceforth in $N_a$.

We show by contradiction that $v$ is accepting infinitely often.
For $\lambda(v)=a$, the proof is essentially the same as for one-Rabin determinisation.
For $\lambda(v)>a$, the proof is altered by a case distinction, where one case assumes that, for some index $i >0$ such that, for all $j \geq i$, $v$ is a prefix of all $v_j$, $(q_{j-1},\alpha(j-1),q_j)\in N_a$, and $(d_{j-1},\alpha(j-1),d_j)\notin R_v \cup A_v$, $q_i$ is in the label of a natural child $vc$ of $v$.
This provides the induction basis -- in the one-pair Rabin case, the basis is provided through the accepting transition of the one-pair Rabin automaton, and we have no corresponding transition with even priority $\geq \lambda(v)$ -- by definition.
If no such $i$ exists, we choose an $i$ that satisfies the above requirements except that $q_i$ is in the label of a natural child $vc$ of $v$.
We can then infer that the label of $v\mathfrak{s}$ also henceforth contains $q_i$.
As a Rabin root whose parent is not accepting or rejecting, $v\mathfrak{s}$ is not rejecting either.
\end{proofidea}

\begin{proof} We fix an accepting run $\rho=q_{0}q_{1} \ldots$ of $\mathcal{P}$ on an input word $\alpha$, and use $a = \liminf_{i\rightarrow \infty}\big((q_i,\alpha(i),q_{i+1})\big)$ to refer to the dominating even priority of its transitions $\overline{\rho}$.
We also let $\rho_{\mathcal DR}=d_{0}d_{1}\ldots$ be the run of $\mathcal DR$ on $\alpha$.
We then define the related sequences of host nodes $\vartheta=v_{0}v_{1}v_{2}\ldots=\host_a(q_0,d_0)\host_a(q_1,d_1)\host_a(q_2,d_2)\ldots$.

Note that Rabin roots cannot be the $a$ host node of any state, as it is always labelled by the union of the labels of its children, and its children have the same level as the Rabin root itself.

Let
\begin{itemize}
\item $v'$ be the longest sequence, which is the initial sequence of almost all $v_i$, and
\item $v$ the longest initial sequence of $v'$, such that, for no initial sequence $v''$ of $v$ (including $v$ itself), infinitely many transitions $(d_i,\alpha(i),d_{i+1})$ are in $R_a$. 
\end{itemize}

We first observe that such a node $v$ exists: as $q_i \in l_i(\varepsilon)$ for $d_i=(\T_i,l_i,\lambda_i)$ for all $i \in \omega$, $\varepsilon$ satisfies all requirements except for maximality, such that a maximal element $v$ exists.
We now distinguish two cases.

`$a = \lambda(v)$':
The first case is that the level of the node $v$ equals the dominating priority of $\overline{\rho}$.
For this case, we can argue as in the one-Rabin pair case:
if the transition is infinitely often in the set $A_v$ of $\mathcal{DR}$, then $\rho_{\mathcal{DR}}$ is accepting.
Otherwise we choose a point $i \in \omega$ with the following properties:
\begin{itemize}
 \item for all $j \geq i$, $(q_j,\alpha(j),q_{j+1}) \in N_a$,
 \item for all $j \geq i$ and all initial sequences $w$ of $v$, $(d_j,\alpha(j),d_{j+1}) \notin R_w$,
 \item for all $j \geq i$, $(d_j,\alpha(j),d_{j+1}) \notin A_v$, and
 \item $\pri(q_i,\alpha(i),q_{i+1})=a$.
\end{itemize}

We can now build a simple inductive argument with the following ingredients.
\begin{description}
 \item[Induction basis:] $ $ \newline
 There is a $k \in \omega$ such that $q_{i+1} \in l_{i+1}(vk)$.
 
 The induction basis holds as the transition $(q_i,\alpha(i),q_{i+1})$ is in $A_a$ and the node $v$ is stable and non-accepting in $(d_i,\alpha(i),d_{i+1})$.
 
\item[Induction step:]  $ $ \newline
if, for some $k \in \omega$ and $j > i$, $q_{j} \in l_{j}(vk)$, then
\begin{itemize}
 \item there is a $k'\leq k$ such that $q_{j+1} \in l_{j+1}(vk')$, and
 \item if $k=k'$ then $(d_j,\alpha(j),d_{j+1}) \notin R_{vk}$.
\end{itemize}

To see this, $q_{j+1}$ is added to the `$l_1(vk)$' from Step 1 of the transition mechanism of the transition $(d_j,\alpha(j),d_{j+1})$.
As $v$ is stable but not accepting, the two only reason for $q_{j+1} \notin l_{j+1}(vk)$ are that
\begin{itemize}
 \item there is, for some $k'' < k$, a $q \in l_j(vk'')$ with and $(q,\alpha(j),q_{j+1}) \in N_{\lambda(v)}$ (note that $\lambda(v) = \lambda(vk) = \lambda(vk'') = a$), or
 \item for some $k'' < k$, the node $vk''$ is removed in Step 5 of the transition mechanism of the transition $(d_j,\alpha(j),d_{j+1})$.
\end{itemize}
In both cases (and their combination), we have $k' < k$.
If neither is the case, then $(d_j,\alpha(j),d_{j+1}) \notin R_{vk}$ (as $\rename(vk)=vk$ holds in the transition mechanism). 
\end{description}

The position $k \in \omega$ of the child $vk$ with $q_j \in l(vk)$ can thus only be decreased finitely many times (and $\lambda(vk)=a$ for all $k \in \omega$).
For some $k\in \omega$, $vk$ is therefore a prefix of almost all $v_i$ of $\vartheta$.
Once stable, it is henceforth no more rejecting.
This contradicts the assumption that $v$ is the longest such sequence.

`$a > \lambda(v)$':
The second case is that the level of $v$ is strictly greater than the dominating priority of $\overline{\rho}$.
We argue along similar lines.
If the transition is infinitely often in the set $A_v$ of $\mathcal{DR}$, then $\rho_{\mathcal{DR}}$ is accepting.
Otherwise we choose a point $i \in \omega$ with the following properties:
\begin{itemize}
 \item for all $j \geq i$, $(q_j,\alpha(j),q_{j+1}) \in N_a$,,
 \item for all $j \geq i$ and all initial sequences $w$ of $v$, $(d_j,\alpha(j),d_{j+1}) \notin R_w$, and
 \item for all $j \geq i$, $(d_j,\alpha(j),d_{j+1}) \notin A_v$.
\end{itemize}

The difference to the previous argument is that the third prerequisite, `$\pri(q_i,\alpha(i),q_{i+1})=a$', holds no longer.
This was used for the induction basis.
We replace this by a distinction of two sub-cases.

The first one is, that we do have an induction basis: we can choose the $i$ such that there is a $k \in \omega$ such that $q_{i+1} \in l_{i+1}(vk)$.
The rest of the argument can be copied for this case:
\begin{description}
\item[Induction step:]  $ $ \newline
if, for some $k \in \omega$ and $j > i$, $q_{j} \in l_{j}(vk)$, then
\begin{itemize}
 \item there is a $k'\leq k$ such that $q_{j+1} \in l_{j+1}(vk')$, and
 \item if $k=k'$ then $(d_j,\alpha(j),d_{j+1}) \notin R_{vk}$.
\end{itemize}

To see this, $q_{j+1}$ is added to the `$l_1(vk)$' from Step 1 of the transition mechanism of the transition $(d_j,\alpha(j),d_{j+1})$.
As $v$ is stable but not accepting, the two only reason for $q_{j+1} \notin l_{j+1}(vk)$ are that
\begin{itemize}
 \item there is, for some $k'' < k$, a $q \in l_j(vk'')$ with and $(q,\alpha(j),q_{j+1}) \in N_{\lambda(v)}$ (note that $\lambda(v) = \lambda(vk) = \lambda(vk'') = a$), or
 \item for some $k'' < k$, the node $vk''$ is removed in Step 5 of the transition mechanism of the transition $(d_j,\alpha(j),d_{j+1})$.
\end{itemize}
In both cases (and their combination), we have $k' < k$.
If neither is the case, then $(d_j,\alpha(j),d_{j+1}) \notin R_{vk}$ (as $\rename(vk)=vk$ holds in the transition mechanism). 
\end{description}

The position $k \in \omega$ of the child $vk$ with $q_j \in l(vk)$ can thus only be decreased finitely many times (and $\lambda(vk)=a$ for all $k \in \omega$).
For some $k\in \omega$, $vk$ is therefore a prefix of almost all $v_i$ of $\vartheta$.
Once stable, it is henceforth no more rejecting.
This contradicts the assumption that $v$ is the longest such sequence.

The other sub-case is that no such $i$ exists.
We then choose $i$ such that the two remaining conditions are met. 
As $\lambda(v)>a \geq 2$ holds, the union of the labels of the children of $v$ must be the same as the label of $v$.
Consequently, we have $q_j \in l(v\mathfrak{s})$ for all $j > i$.
It remains to show that $v\mathfrak{s}$ is not rejecting infinitely many times.
But the only ways a Rabin root can be rejecting is that its parent node is accepting (the breakpoint of Step 4 from the transition mechanism) or not stable (Step 5 with Step 3, removing states from the label that occur in younger siblings) in a transition.
But both are excluded in the definition of $i$.

Finally, we note that, for all $v_i$ in $\vartheta$, $\lambda(v_i) = a$ holds by construction.
Consequently, $\lambda(v_i') \geq a$ holds for all initial sequences $v_i'$ of $v_i$.
In particular, we have $\lambda(v) \geq a$, such that the above case distinction is complete.
\end{proof} 

\begin{lemma} $L(\mathcal{DR}) \subseteq L(\mathcal{P})$
\label{deterp2}
\end{lemma}

The proof of this lemma is essentially the proof of Lemma \ref{lemma:determ2} where, for the priority $a=\lambda(v)$ chosen to be the level of the accepting index $v$, $A_a$ takes the role of the accepting set $A$ from the one-pair Rabin automaton.

\paragraph*{Notation.}
We denote with $Q_1 \Rightarrow^\alpha_a Q_2$ for a finite word $\alpha=\alpha_1\ldots\alpha_{j-1}$ that there is, for all $q_j\in Q_2$, a sequence
$q_1\ldots q_j$ with
\begin{itemize}
 \item $q_1 \in Q_1$,
 \item $(q_i,\alpha_i,q_{i+1})\in N_a$ for all $1\leq i <j$, and
 \item $(q_i,\alpha_i,q_{i+1})\in A_a$ for some $1\leq i <j$.
\end{itemize}

\begin{proof}
Let $\alpha \in L(\mathcal{DR})$.
Then there is a $v$ that is eventually always stable and always eventually accepting in the run $\rho_{\mathcal DR}$ of $\mathcal DR$ on $\alpha$.
We pick such a $v$.

Let $1 < i_{0}<i_{1}<i_{2}<\ldots$ be an infinite ascending chain of indices such that
\begin{itemize}
 \item $v$ is stable for all transitions $(d_{j-1},\alpha(j-1),d_{j})$ with $j\geq i_0$, and
 \item the chain $ i_{0}<i_{1}<i_{2}<\ldots$ contains exactly those indices $i\geq i_0$ such that $(d_{i-1},\alpha(i-1),d_{i})$ is accepting.
\end{itemize}

Let $d_i=(\T_i,l_i,\lambda_i)$ for all $i \in \omega$.
By construction, we have
\begin{itemize}
 \item $I \rightarrow^{\alpha[0,i_0[} l_{i_0}(v)$, and
 \item $l_{i_j}(v) \Rightarrow^{\alpha[i_j,i_{j+1}[}_a l_{i_{j+1}}(v)$.
\end{itemize}

Using this observation, we can build a tree of initial sequences of runs as follows:
we build a tree of initial sequences of runs of $\mathcal P$ that contains a sequence
$q_0q_1q_2\ldots q_{i_j}$ for any $j\in \omega$ iff
\begin{itemize}
 \item $(q_i,\alpha(i),q_{i+1}) \in T$ is a transition of $\mathcal P$ for all $i < i_j$,
 \item $(q_i,\alpha(i),q_{i+1}) \in N_a$ is not rejecting for all $i \geq i_0 - 1$, and
 \item for all $k<j$ there is an $i\in [i_k,i_{k+1}[$ such that $(q_i,\alpha(i),q_{i+1}) \in A_a$ is an accepting transition.
\end{itemize}

By construction, this tree has the following properties:
\begin{itemize}
 \item it is infinite,
 \item it is finitely branching,
 \item no branch contains more than $i_0$ transitions with odd priority $>a$, and,
 \item for all $j \in \omega$, a branch of length $>i_j$ contains at least $j$ transitions with even priority $\geq a$.
\end{itemize}

Exploiting K\"onig's lemma, the first two properties provide us with an infinite path, which is a run of $\mathcal P$ on $\alpha$.
The last two properties then imply that this run is accepting. 
$\alpha$ is therefore in the language of $\mathcal P$.
\end{proof}

\begin{corollary}
\label{parityeq}
$L(\mathcal{P}) = L(\mathcal{DR})$.
\end{corollary}

\subsection{Determinising to a deterministic parity automata $\mathcal D$}
\label{subs:det2parity}
Deterministic parity automata seem to be a nice target when determinising parity or one-pair Rabin automata given that algorithms that solve parity games (e.g, for acceptance games of alternating and emptiness games of nondeterministic parity tree automata) have a lower complexity when compared to solving Rabin games. For \buchi\ and Streett automata, determinisation to parity automata was first shown by Piterman in \cite{Piterman/07/Parity}. For applications that involve co-determinisation, the parity condition also avoids the intermediate Streett condition.

Safra's determinisation construction (and younger variants) intuitively enforces a parity-like order on the nodes of history trees. By storing the order in which nodes are introduced during the construction, we can capture the Index Appearance Records construction that is traditionally used to convert Rabin or Streett automata to parity automata. To achieve this, we augment the states of the deterministic automaton (RHTs or NHTs) with a \emph{later introduction record} (LIR), an abstraction of the order in which the non-Rabin nodes of the ordered trees are introduced. (As Rabin roots are but redundant information, they are omitted in this representation.)

For an ordered tree $\mathcal{T}$ with $m$ nodes that are no Rabin roots, an LIR is a sequence $v_{1}, v_{2},\dots v_{m}$ that contains the nodes  of $\mathcal{T}$ that are no Rabin roots nodes, such that, each node appears after its ancestors and older siblings. 
For convenience in the lower bound proof, we represent a node $v \in \T$ of an NHT $d=(\T,l,\lambda)$ in the LIR by a triple $(S_\p,c_\p,P_\p)$ where $S_v = l(v)$, is the label of $v$, $c_v = \lambda(v)$ the level of $v$, and $P_v = \{q \in Q \mid v=\host_{c_v}(q,d)\}$ is the set of states $c_v$ hosted by $v$.
The $v$ can be reconstructed by the order and level. 
We call the possible sequences of these triples \emph{LIR-NHTs}.
Obviously, each LIR-NHT defines an NHT, but not the other way round.

A finite sequence $(S_1,c_1,P_1) (S_2,c_2,P_2)(S_3,c_3,P_3) \ldots \linebreak[2] (S_k,c_k,P_k) $ of triples is a LIR-NHT if it satisfies the following requirements for all $i \in [k]$.
\begin{enumerate}
 \item $P_i \subseteq S_i$,
 \item $\{P_i\} \cup \{S_j \mid j{>}i,\  c_i {=} c_j,$ and $S_j {\cap} S_i {\neq} \emptyset\}$ partitions~$S_i$.
 \item $\{S_j \mid j>i,\  c_i = c_j+2,$ and $S_j \cap P_i \neq \emptyset\}$ partition $P_i$.
 \item If the highest priority of $\mathcal P$ is even, then $c_i = e$ implies $S_i \subseteq S_1$. (In this case, the lowest level construction is B\"uchi and the first triple always refers to the root.)
 \item For $c_i < e$, there is a $j<i$ with $S_i \subseteq P_j$.
\end{enumerate}

To define the transitions of $\mathcal D$, we can work in two steps.
First, we identify, for each position $i$ of a state $N = (S_1,c_1,P_1) (S_2,c_2,P_2)(S_3,c_3,P_3) \ldots$ of $\mathcal D$, the node $v_i$ of the NHT $d=(\T,l,\lambda)$ for the same input letter.
We then perform the transition $\big(d,\sigma,(\T',l',\lambda')\big)$ on this Rabin automaton.
We are then first interested in the set of non-rejecting nodes from this transition and their indices.
These indices are moved to the left, otherwise maintaining their order.
All remaining vertices of $\T'$ are added at the right, maintaining orderedness.

The priority of the transition is determined by the smallest position $i$ in the sequence, where the related node in the underlying tree is accepting or rejecting.
It is therefore more convenient to use a min-parity condition, where the parity of $\liminf_{n \rightarrow \infty} \pri(\overline{\rho})$ determines acceptance of a run $\rho$.
As this means smaller numbers have higher priority, $\pri$ is representing the opposite of a priority function, and we refer to the priority as the \emph{co-priority} for clear distinction.

If the smallest node is rejecting, the transition has co-priority $2i-1$, if it is accepting (and not rejecting), then the transition has co-priority $2i$, and if no such node exists, then the transition has co-priority $ne+1$.

\begin{lemma}
\label{parityexp}
Given a nondeterministic parity automaton $\mathcal{P}$ with $|P| = n$ states and maximal priority $c$, we can construct a language equivalent deterministic parity automaton $\mathcal D$ with $n e +1$ priorities for $e=2\lfloor 0.5c \rfloor$, whose states are the LIR-NHTs described above.
\end{lemma}

\begin{proof}
We use our determinisation technique from \textbf{Section~\ref{paritydet}} to construct a deterministic parity automaton, whose states consist of the LIR-NHTs, i.e., the NHTs augmented with the Later Introduction Records, with the parity index on the transitions from the states of the automata.

First, we observe that $\mathcal P$ is language equivalent to the deterministic Rabin automaton $\mathcal{DR}$ from the construction of Section \ref{paritydet} by Corollary \ref{parityeq}.

Let $\alpha$ be a word in the language $L(\mathcal{DR})$ of the automaton $\mathcal{DR}$. 
By definition of acceptance, we have an index $v$ such that the node $v$ is a node, which is eventually always stable and always eventually accepting in the transitions of the run of $\mathcal{DR}$ on $\alpha$.
Note that $v$ cannot be a Rabin root, as Rabin roots cannot be accepting.

Once stable, the position of this node in the LIR is non-increasing, and it decreases exactly when a node at a smaller position is deleted. This can obviously happen only finitely many times, and the position will thus eventually stabilise at some position $p$.
Moreover, all positions $\leq p$ will then be henceforth stable.

Then, by our construction, it is easy to see that henceforth no transition can have a co-priority $<2p$.
At the same time, for each following transition where $v$ is accepting in the deterministic Rabin automaton, the respective transition of the run of $\mathcal P$ has a priority $\leq 2p$. (At some node that is represented in a position $\leq 2p$, an accepting or rejecting event happens.) These two observations provide, together with the fact that these priorities $\leq 2p$ must occur infinitely many times by the deterministic Rabin automaton being accepting, that the dominating priority of the run is an even priority $\leq 2p$.

In the other direction, let $2i$ be the dominant priority for a run of our DPA $\mathcal D$ on a word $\alpha$. This leads to a scenario where all positions $\leq i$ eventually maintain their positions in the LIR.
The respective nodes they represent remain stable, but not accepting, from then on in the transitions of the run of $\mathcal{DR}$ on $\alpha$.

Observe that all older siblings (and ancestors, except for the omitted Rabin root) of a node $v$ of an NHT are represented on a smaller position than $v$.
The node corresponding to the position $i$ is always eventually accepting in the transitions of $\mathcal{DR}$ on $\alpha$, such that $\alpha$ is accepted by $\mathcal{DR}$.
\end{proof}

\begin{lemma}
\label{parityref}
The DPA resulting from determinising a one-pair Rabin automaton $\mathcal{R}_1$ has $O(n!^2)$ states, and $O\big(n!(n-1!)\big)$ if $\mathcal{R}_1$ is B\"uchi.
\end{lemma}

\begin{proof}
Let $|Q| = n$ be the number of states of our nondeterministic one-pair Rabin automaton. We explicitly represent (for the sake of evaluating the state-space) the tree structure of an RHT/LIR pair with $m$ nodes by a sequence of $m-1$ integers $i_1,i_2 \dots i_m$ such that $i_j$ points to the position $<j$ of the parent of the node $v_j$ in the LIR $v_{1}, v_{2},\dots v_{m}$. There are $(m-1)!$ such sequences. There is an obvious bijection between this representation of an LIR and its original definition. Thus, for an RHT/LIR pair with $n+1$ nodes, we can have upto $n!$ such RHT/LIR pairs just by virtue of the order of introduction of the nodes.

To more accurately evaluate the number of states, we have to consider the way RHTs are labelled. The root is always labelled with the complete set of reachable states.

We first consider the case where the root is labelled with all $|n|$ states of the nondeterministic one-pair Rabin automaton $\mathcal{R}_1$. Let $t(n,m)$ denote the number of trees and later introduction record pairs for history trees with $m$ nodes and $n=|Q|$ states in the label of the root. First, $t(n,n+1) = (n!\cdot n!)$ holds : For such a tree, there can be upto $n!$ onto functions that resemble the labelling of states of the deterministic automaton and $n!$ RHTs augmented with LIRs. For every $m\leq (n+1)$, a coarse estimation%
\footnote{If we connect functions by letting a function $g$ from $Q$ onto $\{1,\ldots, m-1\}$ be the successor of a function $f$ from $Q$ onto $\{1,\ldots, m\}$ if there is an index $i \in \{1,\ldots, m-1\}$ such that $g(q)=i$ if $f(q)= m$ and $g(q)=f(q)$ otherwise, then the functions onto $m$ have $(m-1)$ successors, while every function onto $m-1$ has at least two predecessors.
Hence, the number of labelling functions grows at most by a factor of $\frac{m-1}{2}$, while the number of ordered tree / LIR pairs is reduced by a factor of $m-1$.}
provides $t(n,m-1)\leq \frac{1}{2}t(n,m)$. Hence, $\sum_{i=1}^n t(n,i)\leq 2 (n!\cdot n!)$.

We next consider the case where the root is not labelled with all $|n|$ states of the nondeterministic one-pair Rabin automaton $\mathcal{R}_1$. Let $t'(n,m)$ denote the number of history tree / LIR pairs for such history trees with $m$ nodes for a nondeterministic one-pair Rabin automaton with $n$ states.
We have $t'(n,n) = (n-1)!n!$ and, by an argument similar to the one used in the analysis of $t$, we also have $t'(n,m-1)\leq \frac{1}{2}t'(n,m)$ for every $m \leq n$, and hence $\sum_{i=1}^{n-1} t'(n,i)\leq 2 (n-1)!n!$.

Overall, the number of RHTs augmented with LIRs is $\sum_{i=1}^{n}{t(n,i)} + \sum_{i=1}^{n-1}{t'(n,i)} \le O(n!^2)$. The number of states of the resulting deterministic parity automaton is $O(n!^2)$, which equates to a linear increase in size when compared with a deterministic parity automaton resulting from the determinisation of a language equivalent nondeterministic \buchi\ automaton instead of a nondeterministic one-pair Rabin automaton.%
\footnote{A similar estimation for the case of determinising \buchi\ automata to parity automata would result in $O\big((n-1)!n!\big)$ states, when the acceptance condition is placed on the transitions rather than the states.} 
\end{proof}
 
\section{Lower Bound}
In this section, we establish the optimality of our determinisation to Rabin automata, and show that our determinisation to parity automata is optimal up to a small constant factor.
What is more, this lower bound extends to the more liberal Streett acceptance condition.

The technique we employ is similar to \cite{Colcombet+Zdanowski/09/Buchi,Schewe+Varghese/12/GBA}, in that we use the states (for Rabin automata) or a large share of the states (for parity automata) of the resulting automaton as memory in a game, and argue that it can be won, but not with less memory.
Just as in \cite{Colcombet+Zdanowski/09/Buchi,Schewe+Varghese/12/GBA}, we use a game where this memory is a lower bound for the size of a deterministic Rabin automaton that recognises the language of a full nondeterministic automaton (see below).
To estimate the size of the minimal Streett automaton, we use the complement language instead.
Consequently, we get a dual result:
a lower bound for a Rabin automaton that recognises the complement language.
By duality, this bound is also the lower bound for a deterministic Streett automaton that recognises the language of this full automaton.
As the parity condition is a special Streett condition, this lower bound extends to parity automata.

\subsection{Full automata}
Our lower bound proof builds on \emph{full} automata (cf.\ \cite{Yan/08/lowerComplexity}), like the ones used in \cite{Colcombet+Zdanowski/09/Buchi} to establish a lower bound for the translation from nondeterministic B\"uchi to deterministic Rabin automata.
%
%
A parity automaton  $\mathcal P_n^c = \big(Q,\Sigma_n^c,I,T,\pri\big)$ with $n$ states is called \emph{full} if
its alphabet $\Sigma_n^c= Q \times Q^\top \rightarrow 2^\nc$ is the set of functions from $Q\times Q^\top$ to sets of priorities $\nc$, and
\begin{itemize}
 \item $I=Q$,
 \item $T = \big\{\big(q,\sigma,q') \mid  q \in Q,\ q'\in Q^\top,\ \sigma(q,q') \neq \emptyset\big\}$,
 \item $\pri: (q,\sigma,q') \mapsto \opt\big(\sigma(q,q')\big)$ for all $q,q' \in Q$ with $\sigma(q,q')\neq \emptyset$, where $\opt$ returns the highest even number of a set, and the lowest odd number if the set contains no even numbers.
\end{itemize}

$(q,\sigma,\top)$ encodes immediate acceptance from state $q$. 
Every nondeterministic parity automaton with priorities $\leq c$ can be viewed as a language restriction (by alphabet restriction) of $\mathcal P_n^c$.
$\mathcal P_n^c$ therefore recognises the hardest language recognisable by a parity automata with $n$ states and maximal priority $c$.

To estimate the size of deterministic Rabin, Streett, or parity automata that recognise the same language as a nondeterministic parity automaton with $n$ states and maximal priority $c$ reduces to estimating the size of the deterministic Rabin, Streett, or parity automata that recognises the language of $\mathcal P_n^c$.
A useful property of this language is that we can focus on states with different sets of reachable states independently.
We use $\reach(u)$ to denote the states reachable by a word $u \in \Sigma^*$ in $\mathcal P_n^c$ that is not immediately accepted by $\mathcal P_n^c$ (that is, such that $\top$ is not reachable on $u$).
$\reach(u)$ can be defined inductively: $\reach(\varepsilon)=I$, and $\reach(va) = \big\{q' \in Q \mid \exists q \in \reach(v).\ (q,\sigma,q') \in T\big\}$ for all words  $v\in \Sigma^*$ and all letters $a \in \Sigma$.
This allows us to extend a useful observation from \cite{Colcombet+Zdanowski/09/Buchi}.

\paragraph*{Notation.}
In this section, we use $\rho(s,u)$ to refer to a finite part of the run of a deterministic automaton that starts in a state $s$ upon reading a word $u \in {\Sigma_n^c}^+$.
If this finite part of the run ends in a state $s'$, we also write $\rho(s,u,s')$.
In particular, $\rho(s,u,s')$ implies that $s'$ is reached from $s$ when reading $u$.
For Rabin automata, an index $i$ is accepting resp.\ rejecting for $\rho(s,u,s')$, if it is accepting resp.\ rejecting in some transition in this sequence of a run.
For parity automata, the co-priority of $\rho(s,u,s')$ is the smallest co-priority that occurs in any transition in the respective sequence of a run.

\begin{lemma}
\label{lem:differentSets}
Let $\mathcal A$ be a deterministic Rabin, Street, or parity (or, more generally, Muller) automaton that recognises the language of $\mathcal P_n^c$ or its complement.
Then $\rho(s_0,u,s)$ and $\rho(s_0,v,s)$ imply $\reach(u) = \reach(v)$ or $\reach(u) \ni \top \in \reach(v)$.
\end{lemma}

\begin{proof}
Assume for contradiction that this is not the case.
We select two words $u,v \in \Sigma^*$ with $\rho(s_0,u,s)$ and $\rho(s_0,v,s)$.
Let $\sigma_\emptyset : (q,q') \mapsto \emptyset \forall (q,q')\in Q\times Q^\top$.

If $\reach(u) \ni \top \notin \reach(v)$, then $v {\sigma_\emptyset}^\omega$ is accepted, and $u {\sigma_\emptyset}^\omega$ is rejected by $\mathcal P_n^c$.

If $\top \notin \reach(u) \cup \reach(v)$ and $q \in \reach(u)\smallsetminus \reach(v)$, then we use $\sigma_q : (q,q)  \mapsto \{2\}$ and $\sigma_q : (q',q'')  \mapsto \emptyset$ if $(q',q'') \neq (q,q)$.
Then $u {\sigma_q}^\omega$ is accepted, while $ v {\sigma_q}^\omega$ is rejected by $P_n^c$.

Doing the same with $u$ and $v$ reversed provides us with the required contradiction.
\end{proof}

As a consequence, we can focus on sets of states with the same reachability set.

\subsection{Language games}
A language game is an initialised two player game $G=(V,E,v_0,\mathcal L)$, which is played between a verifier and a spoiler on a star-shaped directed labelled multi-graph $(V,E)$ without self-loops.
It has a finite set $V$ of vertices, but a potentially infinite set of edges.

The centre of the star, which we refer to by $c \in V$, is the only vertex of the verifier, while all other vertices are owned by the spoiler.
Besides the centre, the game has a second distinguished vertex, the initial vertex $v_0$, where a play of the game starts.
The remaining vertices $W = V \smallsetminus \{v_0,c\}$ are called the working vertices.
Like $v_0$, they are owned by the spoiler.

The edges are labelled by finite words over an alphabet $\Sigma$.
Edges leaving the centre vertex are labelled by the empty word $\varepsilon$, and there is exactly one edge leaving from the edge to each working vertex, and no outgoing edge to the initial vertex.
The set of these outgoing edges is thus $\{(c,\varepsilon,v)\mid v \in W\}$.
The edges that lead to the centre vertex are labelled with non-empty words.

The players play out a run of the game in the usual way by placing a pebble on the initial vertex $v_0$, letting the owner of that vertex select an outgoing edge, moving the pebble along it, and so forth.
This way, an infinite sequence of edges is produced, and concatenating the finite words by which they are labelled provides an infinite word $w$ over $\Sigma$.
The verifier has the objective to construct a word in $\mathcal L$, while the spoiler has the antagonistic objective to construct a word in $\Sigma^\omega \smallsetminus \mathcal L$.

\begin{theorem}
\cite{Colcombet+Zdanowski/09/Buchi}
\label{theo:Lgame}
If the verifier wins a language game for a language recognised by a deterministic Rabin automaton $\mathcal R$ with $r$ states, then he wins the language game using a strategy with memory $r$.
\end{theorem}

This is because he can simply run $\mathcal R$ as a witness automaton.
Intuitively, the verifier would play on the product of $\mathcal R$ and $G$.
This is a Rabin game, and if the verifier wins, then he wins memoryless \cite{DBLP:journals/apal/Klarlund94,Zielonka/98/Parity}.
Thus, the states of $\mathcal R$ can serve as the memory in $G$:
the verifier will simply make the decision defined by the decision he made in the product game.

\begin{corollary}
\label{cor:Lgame}
If the verifier wins a language game for a language recognised by a deterministic Rabin automaton $\mathcal R$ with $r < |W|$ states, then he wins the language game played on a reduced graph, where the set of his outgoing edges is reduced to $r$ edges of his choice before playing the otherwise unchanged game.
\end{corollary}

These are simply the at most $r$ edges chosen by the verifier under the at most $r$ different memory states.

\subsection{Lower bounds}
We extend the technique introduced by Colcombet and Zdanowski \cite{Colcombet+Zdanowski/09/Buchi} to establish that the Rabin automata from Corollary \ref{parityeq} are the minimal deterministic Rabin automata that recognise the same language as $\mathcal P_n^c$.
Just as in \cite{Colcombet+Zdanowski/09/Buchi,Schewe+Varghese/12/GBA}, we use the language of $\mathcal P_n^c$ as the target language.

To establish that the deterministic parity automaton $\mathcal D_n^c$ from Lemma \ref{parityexp} cannot be $50\%$ larger than any deterministic Streett -- and thus in particular than any deterministic parity -- automaton that recognises the language of $\mathcal D_n^c$, we use the complement language of $\mathcal P_n^c$ as our target language.

We therefore get a bound on the size of the smallest Rabin automaton that recognises the complement of the language of $\mathcal P_n^c$, and hence for the smallest Streett automaton that recognises $\mathcal P_n^c$.
Having an upper bound for parity that matches this lower bound for the more general Streett condition, we can infer tightness of our determinisation construction for both classes of automata.

\paragraph*{\bf Deterministic Rabin automata.}
To establish the lower bound, it is easier to use triples $(S_v,c_v,P_v)$ for each node $v \in \T$ of an NHT $(\T,l,\lambda)$.
By abuse of notation, we refer to the triple by $\T(v)$ (and thus to the state of the DRA by $\T$), to label by $T_S(v)$ and to $\{q {\,\in\,} Q \mid v=\host_{\lambda(v)}(q,\T)\}$ by~$\T_P(v)$.

To define the edges leaving a spoiler vertex $\T$, we refer to the finite part of a run of $\mathcal R_n^c$ that starts in $\T$ when reading a word $u \in {\Sigma_n^c}^+$ by $\rho(\T,u)$.
If this finite part of the run ends in $\T'$, we also write $\rho(\T,u,\T')$.
In particular, $\rho(\T,u,\T')$ implies that $\T'$ is reached from $\T$ when reading $u$.
The accepting and rejecting nodes of $\rho(\T,u,\T')$ are the union of the accepting and rejecting nodes, respectively, of the individual transitions in this section of the run.

\begin{definition}
[Relevant change]
In a finite part $\rho(\T,u,\T')$, of our Rabin automaton $\mathcal R_n^c$ the \emph{relevant change} is the minimal position $\p$ w.r.t.\ lexicographic order, where
\begin{itemize}
 \item the node has been accepting or rejecting during the piece of the run, or
 \item where $\T(\p)\neq \T'(\p)$.
\end{itemize}
We call the node $\p$ the relevant change, and we call it
\begin{itemize}
 \item \emph{rejecting}, if $\rho(\T,u,\T')$ is rejecting at $\p$,
 \item \emph{accepting}, if $\rho(\T,u,\T')$ is accepting but not rejecting at $\p$,
 \item \emph{growing}, if it is not rejecting and $\T_S'(\p)\supsetneq \T_S(\p)$, and
 \item \emph{shrinking}, if it is not rejecting, $\T_S'(\p) = \T_S(\p)$ and $\T_P'(\p) \subsetneq \T_P(\p)$.
\end{itemize}
\end{definition}
 
We use a set of language games, one for each subset $S\subseteq Q$ of the states $Q$ of $\mathcal P_n^c$ with two or more states.
The vertices of such a language game consist of the centre vertex, the initial vertex, and the working states $W$.
These working states consist of the states of $\mathcal R_n^c$ with $\reach(\T)=S$.
The target language is the language of all words \emph{accepted} by $\mathcal R_n^c$, and we have the following edges:

\begin{itemize}
\item there is an edge $(v_0,u,c)$ for all $u \in {\Sigma_n^c}^+$ with $\rho(\T_0,u,\T)$ and $\T \in W$,
\item $(c,\varepsilon,\T)$ for all $\T \in W$, and
\item $(\T,u,c)$ if $\rho(\T,u,\T')$ is accepting, growing, or shrinking, and $\T' \in W$.
\end{itemize}
 
To establish that the minimal Rabin automaton that recognises the language of $\mathcal R_n^c$ cannot be smaller than $R_n^c$, we show that the verifier needs all edges to win each of these games.

\begin{lemma}
\label{lem:rwin}
The verifier wins these language games.
\end{lemma}

For this, we recall the structure of the strategy that the verifier applies: he would use $\mathcal R_n^c$ as a witness automaton, moving to the vertex that represents the state $\T$ that $\mathcal R_n^c$ would be in upon reading the finite word produced so far.
 
If one of his outgoing edges is removed, then there is one such state, say $\T$, he cannot respond to properly.
Instead, he would have to go to a different state $\T'$.
We show that, irrespective of the states $\T$ and $\T' \neq \T$ chosen, the spoiler can produce a word $u\in {\Sigma_n^c}^+$ such that $(\T,u,c)$ is an edge in $G$ and $\rho(\T',u,\T)$ is not accepting in any position.
 
If the spoiler has such an option, then \emph{she} can use $\mathcal R_n^c$ as a witness automaton:
whenever it is her move, 
she chooses an edge with the properties described above.

\begin{proof}
The verifier can simply use the strategy to monitor the state that the monitor DRA $\mathcal R_n^c$ from Corollary \ref{parityeq} would be in.
He then has the winning strategy to play $(c,\varepsilon,\T)$ when the automaton is in state $\T$.
 
To see that he wins the game with this strategy, we consider the run of $\mathcal R_n^c$ on the word defined by the play $(v_0,u_0,c) (c,\varepsilon,\T_1) (\T_1,u_1,c) (c,\varepsilon,\T_2) \ldots $, which refers to the word $u_0u_1u_2\ldots$.
 
The segments $\rho(\T_i,u_i,\T_{i+1})$ of the run have, for all $i \geq 1$, an accepting, growing, or purifying relevant change.
 
Let us consider the relevant changes of these segments.
There is a -- with respect to lexicographic order -- minimal one $\p_{\min}$ that occurs infinitely often.
Let us choose a position $i$ in the play such that no lexicographic smaller $\p'$ is henceforth a relevant change.
Then no node smaller than or equal to $\p$ (with respect to the lexicographic order) can henceforth be rejecting.

Clearly, if $\p$ is infinitely often accepting, then the verifier wins.

Let us assume for contradiction that there is a $j>i$ such that $\p$ is not accepting from position $j$ onwards.
Then, the set of states in $\T_l(v)$ must henceforth grow monotonously with $l$, and grow strictly every time $\p$ is growing.
As this can only happen finitely many times, there is a $k>j$ such that $\p$ is henceforth neither accepting nor growing.

Then, the set of pure states in $\T_l(v)$ must henceforth shrink monotonously with $l$, and shrink strictly every time $\p$ is shrinking.

As this can only happen finitely often, this provides us with the required contradiction.
\end{proof}
 
\begin{lemma}
\label{lem:rwillwin}
Let $\T$ and $\T'$ be two different states of $\mathcal R_n^c$ with $\reach(\T)=\reach(\T')$.
Then there is a word $u\in {\Sigma_n^c}^+$ such that no node in $\rho(\T,u,\T)$ is accepting and $\rho(\T',u,\T)$ has an accepting, growing, or shrinking relevant change.
\end{lemma}
 
\begin{proof}
We first identify the minimal position $\p$ in which $\T$ and $\T'$ are different, and the set $P$ of all lexicographic smaller positions that are part of $\T$ (and thus of $\T'$).
 
We use a word $u = \sigma v$ that consists of an initial letter, in which all nodes but $P$ are rejecting when staring in $\T$ and the nodes in $P$ are neither accepting nor rejecting when staring in $\T$ or $\T'$.
In the second phase, we re-build $\T$ without making any node in $P$ accepting or rejecting.

Now let $\T(\p) = (S_\p,c_\p,P_\p)$, $\T'(\p) = (S_\p',c_\p,P_\p')$, and $\T(\p) = (S_\p,c_\p,P_\p)$ for all $\p' \in P$.
(Recall that the priority is defined by the position in the tree.)

We now distinguish four cases:
(1) there is an $s \in S_\p' \smallsetminus S_\p$, (2) $S_\p' = S_\p$ and there is an $s \in P_\p \smallsetminus P_\p'$, (3)
$S_\p \supsetneq S_\p'$, and (4) $S_\p' = S_\p$ and $P_\p \subsetneq P_\p'$.

We select the first letter of our word as follows.
\begin{enumerate}
\item If there is an $s \in S_\p' \smallsetminus S_\p$, then we can fix such an $s$ and select the first letter of our word as follows:
\begin{itemize}
\item We let $c_\p \in \sigma(s,s')$ for all $s' \in S_\p$.
\item For all $\p' \in P$, all $s' \in P_{\p'}$, and ll $s'' \in S_{\p'}$, we let $c_{\p'}-1 \in \sigma(s',s'')$.
\item If $c$ is odd, we let $c \in \sigma(s',s'')$ for all $s',s'' \in  \reach(\T)$ (in order to maintain the set of reachable states).
\end{itemize}
No further priority is included in any set $\sigma(s',s'')$ for $s',s'' \in Q$ and $s'' \in Q^\top$.

This letter $\sigma$ is chosen such that no node in $P$ is accepting or rejecting in $\delta(\T,\sigma)$ or $\delta(\T',\sigma)$.
While $\p$ is accepting in $\delta(\T',\sigma)$, it is rejecting in $\delta(\T,\sigma)$.
All other positions are rejecting in these transitions.

\item If $S_\p' = S_\p$ and there is an $s \in P_\p \smallsetminus P_\p'$, then we can fix such an $s$ and select the first letter of our word as follows:
 
\begin{itemize}
\item We let $c_\p-1 \in \sigma(s,s')$ for all $s' \in S_p$.
\item For all $\p' \in P$, all $s' \in P_{\p'}$, and ll $s'' \in S_{\p'}$, we let $c_{\p'}-1 \in \sigma(s',s'')$.
\item If $c$ is odd, we let $c \in \sigma(s',s'')$ for all $s',s'' \in  \reach(\T)$ (in order to maintain the set of reachable states).
\end{itemize}
No further priority is included in any set $\sigma(s',s'')$ for $s',s'' \in Q$ and $s'' \in Q^\top$.
 
This letter $\sigma$ is chosen such that no node in $P$ is accepting or rejecting in $\delta(\T,\sigma)$ or $\delta(\T',\sigma)$.
While $\p$ is accepting in $\delta(\T',\sigma)$, it is neither accepting nor rejecting in $\delta(\T,\sigma)$.
All other positions are rejecting in these transitions.

\item If $S_\p \supsetneq S_\p'$, then we can fix an $s \in P_\p$ and select the first letter of our word as follows:
\begin{itemize}
\item We let $c_\p-1 \in \sigma(s,s')$ for all $s' \in S_\p$.
\item For all $\p' \in P$, all $s' \in P_{\p'}$, and ll $s'' \in S_{\p'}$, we let $c_{\p'}-1 \in \sigma(s',s'')$.
\item If $c$ is odd, we let $c \in \sigma(s',s'')$ for all $s',s'' \in  \reach(\T)$ (in order to maintain the set of reachable states).
\end{itemize}
No further priority is included in any set $\sigma(s',s'')$ for $s',s'' \in Q$ and $s'' \in Q^\top$.

This letter $\sigma$ is chosen such that no node in $P$ is accepting or rejecting in $\delta(\T,\sigma)$ or $\delta(\T',\sigma)$.
While $\p$ is not rejecting (but may or may not be accepting) in $\delta(\T',\sigma)$, it is neither accepting nor rejecting in $\delta(\T,\sigma)$.
All other positions are rejecting in these transitions.

\item If $S_\p = S_\p'$ and $P_\p \subseteq P_\p'$, then we can fix an $s \in P_\p$ and select the first letter of our word as follows:
\begin{itemize}
\item We let $c_\p-1 \in \sigma(s,s')$ for all $s' \in S_\p$.
\item For all $\p' \in P$, all $s' \in P_{\p'}$, and ll $s'' \in S_{\p'}$, we let $c_{\p'}-1 \in \sigma(s',s'')$.
\item If $c$ is odd, we let $c \in \sigma(s',s'')$ for all $s',s'' \in  \reach(\T)$ (in order to maintain the set of reachable states).
\end{itemize}
No further priority is included in any set $\sigma(s',s'')$ for $s',s'' \in Q$ and $s'' \in Q^\top$.

This letter $\sigma$ is chosen such that neither $\p$ nor any node in $P$ is accepting or rejecting in $\delta(\T,\sigma)$ or $\delta(\T',\sigma)$.
All other positions are rejecting in these transitions.
\end{enumerate}

Note that $\Delta(\T,\sigma) = \Delta(\T',\sigma)$ holds in all those cases.
Starting with this letter, we can continue to build a word to reconstruct $\T$.
Note that all we have to avoid during this construction is to make $\p$ or a node in $P$ accepting.

The resulting fragments $\rho(\T',u,\T)$ are accepting in cases (1) and (2), growing in case (3), and shrinking in case (4), such that $(\T',u,\T)$ is a transition, while $\rho(\T,u,\T)$ does not contain any accepting node. 
\end{proof}

\begin{corollary}
\label{cor:rlose}
If any outgoing edge is removed from the verifier's centre vertex in any of these games, then the spoiler wins the language game.
\end{corollary}
 
Together with Lemmata \ref{lem:differentSets} and \ref{lem:rwin}, Corollary \ref{cor:rlose} provides:

\begin{theorem}
\label{theo:smallRabin}
$\mathcal R_n^c$ is the smallest deterministic Rabin automaton that recognises the language of $\mathcal P_n^c$.
\end{theorem}

\paragraph*{\bf Deterministic parity automata.}
For our language game, we use the spiked states of $\mathcal D_n^c$ and a fresh initial vertex as the spoiler vertices.
We call a state of $\mathcal D_n^c$ \emph{spiked}, if its last position is a triple of the form $(\{q\},2,\{q\})$.
This is a mild restriction and owed to the fourth case of the proof of Lemma \ref{lem:useOfSpikes}.
Most states are spiked.

\begin{lemma}
\label{lem:spiked}
$\mathcal D_n^c$ has more than twice as many spiked as unspiked states.
\end{lemma}
To define the edges leaving a spoiler vertex $N$, we refer to the finite part of a run of $\mathcal D_n^c$ that starts in $N$ when reading a word $u \in {\Sigma_n^c}^+$ by $\rho(N,u)$.
If this finite part of the run ends in $N'$, we also write $\rho(N,u,N')$.
In particular, $\rho(N,u,N')$ implies that $N'$ is reached from $N$ when reading $u$.
The co-priority of $\rho(N,u,N')$ is the smallest co-priority that occurs in the respective sequence of a run.

\begin{proof}
Each unspiked state ends in a triple $(P,2,P)$ with $|P| \geq 2$.
We can simply replace it by $|P|$ pairs of triples, $(P,2,P\smallsetminus \{q\}),(\{q\},2,\{q\})$, for each $q \in P$.
The resulting state is spiked.
Each non-spiked state produced at least two spiked states, each spiked state is produced by at most one state, and not every spiked state can be produced this way, e.g., states $N$ with $|\reach(N)|=1$ cannot.
\end{proof}


\begin{definition}
[Relevant change]
In 
$\rho(N,u,N')$, the \emph{relevant change} is the minimal position $i$, where
\begin{itemize}
 \item the co-priority of $\rho(N,u,N')$ is $2i-1$ or $2i$  \hfill\mbox{(i.e., position $i$ was accepting or destroyed), or}
 \item the $i$-th position of $N$ and $N'$ differ.
\end{itemize}
For  $N = \big\{(S_j,c_j,P_j)\big\}_{j \leq m}$ and $N' = \big\{(S_j',c_j',P_j')\big\}_{j \leq m'}$, we call the relevant change $i$
\begin{itemize}
 \item \emph{rejecting}, if the co-priority of $\rho(N,u,N')$ is $2i-1$,
 \item \emph{shrinking}, if $S_i' \subsetneq S_i$,
 \item \emph{defying} if $S_i' = S_i$ and the co-priority of $\rho(N,u,N')$ is $2i+1$, and
 \item \emph{purifying}, if $S_i' {=} S_i$, $P_i' {\supsetneq} P_i$, and the co-priority is ${>}2i$.
\end{itemize}
\end{definition}

We use a language game, whose vertices consist of the centre vertex, the initial vertex, and and the working vertices $W$, which form a subset ot the spiked states of $\mathcal D_n^c$.
Following Lemma \ref{lem:differentSets}, we will, for each $S \subseteq Q$ with $|S| \geq 2$, use an individual game where $W$ contains a spiked state $N$ iff $S$ is the set of states reachable in $N$ ($\reach(N)=S$).
The target language is the \emph{complement} language of $\mathcal D_n^c$, and we have the following edges:
\begin{itemize}
\item there is an edge $(v_0,u,c)$ for all $u \in {\Sigma_n^c}^+$ such that there is a spiked state $N \in W$ such that $\rho(N_0,u,N)$, where $N_0$ is the initial state of $\mathcal D_n^c$,
\item $(c,\varepsilon,N)$ for all spiked states $N$ of $\mathcal D_n^c$ that are working states of the game, and
\item $(N,u,c)$ if $\rho(N,u,N')$ is rejecting, shrinking, defying, or purifying, and $N'$ is spiked.
\end{itemize}

\begin{lemma}
The verifier wins all of these language games.
\label{lem:win}
\end{lemma}
\vspace*{-2ex}

\begin{proof}
In the language game for each $S \subseteq Q$, the verifier can simply use the strategy to monitor the state that the monitor DPA $\mathcal D_n^c$ from Lemma \ref{parityexp} would be in.
He then wins by playing $(c,\varepsilon,N)$ when the automaton is in state $N$.

To see that he wins the game, we consider the run of $\mathcal D_n^c$ on the word defined by the play $(v_0,u_0,c) (c,\varepsilon,N_1) (N_1,u_1,c) (c,\varepsilon,N_2) \ldots $, which refers to the word $w=u_0u_1u_2\ldots$.

The run $\rho$ of $\mathcal D_n^c$ on $w$ can be decomposed into the finite segments $\rho(N_i,u_i,N_{i+1})$ for all $i \geq 0$.
For all $i \geq 1$, their relevant changes are rejecting, shrinking, defying, or purifying.

Clearly, there is a minimal one $i_{\min}$ that occurs infinitely often.
Consequently, no co-priority smaller than $2\im-1$ can occur infinitely many times in $\rho$.

We can now distinguish four cases.
\emph{\bf First}, assume that there are infinitely many rejecting relevant changes $\im$.
Then the co-priority $2\im-1$ occurs infinitely often in $\rho$, and the $\omega$ word $w$ is~rejected.

\emph{\bf Second}, assume that finitely many of the relevant changes with change priority $\im$ are rejecting, but infinitely many are shrinking.
Then we can choose a position in the play where henceforth no relevant change with priority $<\im$, and no rejecting relevant change with priority $\im$ occurs.
Consequently, the set of states at position $\im$ of $N_i$ would henceforth shrink monotonously with growing $i$, and would infinitely often shrink strictly. (contradiction)

\emph{\bf Third}, assume that finitely many of the relevant changes with change priority $\im$ are rejecting or shrinking, but infinitely many are defying.
Then we can choose a position in the play where henceforth no relevant change with change priority $<\im$, and no rejecting or shrinking relevant change with change priority $\im$ occurs.
From this time onwards, no co-priority $\leq 2\im$ can occur on any segment of the run, while the co-priority $2\im +1$ occurs infinitely often.

\emph{\bf Finally}, assume that finitely many of the relevant changes with change priority $\im$ are rejecting, shrinking, or defying.
Then we can choose a position in the play where henceforth no relevant change $j$ with $j<\im$, and no rejecting, shrinking, or defying relevant change with relevant change $\im$ occurs.
Consequently, the set of pure states at position $\im$ of $N_i$ would henceforth grow monotonously with growing $i$, and would infinitely often grow strictly. (contradiction)
\end{proof}

To establish that the minimal size of a Rabin automaton that recognises the complement language of $\mathcal D_n^c$ cannot be significantly smaller than $D_n^c$, we will show that the verifier needs all edges to win this game.

For this, we recall the structure of the strategy that the verifier applies: he would use $\mathcal D_n^c$ as a witness automaton, moving to the vertex that represents the state  $\mathcal D_n^c$ would be in upon reading the finite word produced so far.
If one of his outgoing edges is removed, then there is one such state, say $N$, he cannot respond to properly.
Instead, he would have to go to a different state $N'$.

We show that, irrespective of the state $N$ that becomes unreachable and $N' \neq N$ chosen, the spoiler can produce a word $u\in {\Sigma_n^c}^+$ such that $(N',u,c)$ is a transition in $G$ and $\rho(N,u,N)$ has even co-priority.

If the spoiler has such an option, then \emph{she} can use $\mathcal D_n^c$ as a witness automaton:
initially, she selects an edge $(v_0,u,c)$ such that $\rho(N_0,u,N)$ holds; 
henceforth she chooses, whenever she is in a vertex $N$, an edge $(N',u,c)$ such that $\rho(N,u,N)$ holds, 
returning the run to $N$ with dominating even co-priority.
Thus, she can make sure that the constructed word is accepted.

\begin{lemma}
\label{lem:useOfSpikes}
Let $N$ and $N'$ be two different spiked states of $\mathcal D_n^c$ with $\reach(N)=\reach(N')$.
Then there is a word $u\in {\Sigma_n^c}^+$ such that the lowest co-priority occurring in $\rho(N,u,N)$ is even and $\rho(N',u,N)$ has a rejecting, shrinking, defying, or purifying relevant change.
\end{lemma}

\begin{proof}
Let $N = \big\{(S_i,c_i,P_i)\big\}_{i \leq m}$ and $N' = \big\{(S_i',c_i',P_i')\big\}_{i \leq m'}$.
As $N \neq N'$, there is%
\footnote{Note that a spiked state $N'$ cannot simply be longer than a spiked state $N$ with $\reach(N)=\reach(N')$ (or vice versa): assuming that $N$ is an initial sequence of $N'$.
Then the rules (1), (2), (3), and (5) imply that $S_{m+1}'$ must be disjoint with all $S_i$ for $i\leq m$, which contradicts $\reach(N)=\reach(N')$.}
a minimal $\im \leq \min\{m,m'\}$ such that $(S_\im,c_\im,P_\im) \neq (S_\im',c_\im',P_\im')$.

We will construct a word $u$ such that
\begin{itemize}
 \item $\rho(N,u,N)$ and $\rho(N',u,N)$ are fragments of runs,
 \item the minimal co-priority of $\rho(N,u,N)$, is even, and
 \item \im will be the relevant change in $\rho(N',u,N)$; it will be rejecting, shrinking, defying, or purifying.
\end{itemize}

We distinguish four cases.
Let us \emph{\bf first} assume that there is an $s\in S_\im \smallsetminus S_\im'$.
In this case, we choose such an $s$, and select the first letter $\sigma_\im$ of $u$ such that
\begin{itemize}
\item $c_\im\in\sigma_\im(s,s')$ for all $s' \in S_\im$,
\item $c_i-1 \in \sigma_\im(s',s'')$ for all $i < \im$, $s' \in P_i$, and $s'' \in S_i$, and
\item if $c$ is odd%
      \footnote{If the highest priority $c$ of the defining NPA $\mathcal P^c_n$ is odd, then $S_1$ might be a strict subset of $\reach(N)$.
      This part is then an easy way to make sure that all states in $\reach(N)$ remain reachable.
      If $c$ is even, then we have a tree on the lowest level, $S_1=\reach(N)$.},
      $c \in \sigma_\im(s',s'')$ for all $s',s'' \in  \reach(N)$.
\end{itemize}
No further priority is included in any set $\sigma_\im(s',s'')$ for $s',s'' \in Q$ and $s'' \in Q^\top$.

Starting with $\sigma$ is the central step.
The transition from $N$ reading $\sigma$ has co-priority $2\im$,
the transition from $N'$ reading $\sigma$ has co-priority $2\im-1$.
Note that during the transition, all nodes in the history trees underlying $N$ and $N'$ that refer to a position $<\im$ are the same. They are also stable and non-accepting during this transition.
However, while the node the position $\im$ of $N$ refers to is accepting, the node position $\im$ of $N'$ refers to is not stable.
(Note that $N$ and $N'$ could refer to the same underlying tree.)
The resulting state is the same for $N$ and $N'$.

The next letters are to rebuild $N$.
For all $i = \im + 1$ to $m$, we append a further letter $\sigma_i$ to our partially constructed word $u$.
We choose a state $s \in P_i$ and define $\sigma_i$ as follows:
\begin{itemize}
\item $c_i \in \sigma_i(s,s') $ for all $s' \in S_i$,
\item $c_j{-}1 \in \sigma_i(s',s'') $ for all $j{<}i$, $s' {\in} P_j$ and $s'' {\in} S_j$, and
\item if $c$ is odd, $c \in \sigma_\im(s',s'')$ for all $s',s'' \in  \reach(N)$.
\end{itemize}
No further priority is included in any set $\sigma_\im(s',s'')$ for $s',s'' \in Q$ and $s'' \in Q^\top$.

Clearly, reading $i$ from a state that agrees with $N$ on all positions $< i$ before reading $\sigma_i$, the resulting state will agree on all positions $\leq i$ with $N$ after this transition.

The transition has a co-priority $\geq 2i-1 > 2\im$.
Thus, the word $u = \sigma_\im \sigma_{\im+1} \sigma_{\im+2} \ldots \sigma_m$ has the required properties; in particular $\rho(N',u,N)$ has rejecting relevant change $i$.
\smallskip

In the remaining cases we have $S_\im \subseteq S_\im'$.
Note that this implies $c_\im = c_\im'$.

The \emph{\bf second} case is $S_\im = S_\im'$ and there is a state $s \in P_\im' \smallsetminus P_\im$.
In this case, we fix such an $s$ and start our word $u$ with the letter $\sigma_\im$ that satisfies
\begin{itemize}
\item $c_\im-1 \in \sigma_\im(s,s')$ for all $s' \in S_\im$,
\item $c_i-1 \in \sigma_\im(s',s'')$ for all $i < \im$, $s' \in P_i$, and $s'' \in S_i$, and,
\item if $c$ is odd, $c \in \sigma_\im(s',s'')$ for all $s',s'' \in  \reach(N)$.
\end{itemize}
No further priority is included in any set $\sigma_\im(s',s'')$ for $s',s'' \in Q$ and $s'' \in Q^\top$.

Starting $u$ with this letter $\sigma_\im$ is again the central step.
The transition from $N$ has co-priority $2\im$, while the transition from $N'$ has co-priority $2\im+1$.
We can now continue $u$ in the same manner as above and use $u = \sigma_\im \sigma_{\im+1} \sigma_{\im+2} \ldots \sigma_m$, and $u$ will again satisfy the constraints; in particular $\rho(N',u,N)$ has defying relevant change $i$.

In the \emph{\bf third} case, $S_\im \subsetneq S_\im'$, we fix an $s\in S_\im$ and start our word $u$ with the letter $\sigma_\im$ that satisfies
\begin{itemize}
\item $c_\im \in \sigma_\im(s,s')$ for all $s' \in S_\im$,
\item $c_i-1 \in \sigma_\im(s',s'')$ for all $i < \im$, $s' \in P_i$, and $s'' \in S_i$, and,
\item if $c$ is odd, $c \in \sigma_\im(s',s'')$ for all $s',s'' \in  \reach(N)$.
\end{itemize}
No further priority is included in any set $\sigma_\im(s',s'')$ for $s',s'' \in Q$ and $s'' \in Q^\top$.

Starting $u$ with this letter $\sigma_\im$ is again the central step.
The transition from $N$ or $N'$ reading $\sigma$ has co-priority $2\im$.
We can again continue $u$ in the same manner as above and use $u = \sigma_\im \sigma_{\im+1} \sigma_{\im+2} \ldots \sigma_m$, and $u$ will again satisfy the constraints; in particular $\rho(N',u,N)$ has shrinking relevant change $i$.

Finally, in the \emph{\bf fourth} case we have $S_\im = S_\im'$ and $P_\im' \subsetneq P_\im$.
We first note that this implies $|P_\im| \geq 2$.
The restriction to spiked states then provides $\im < m$.
We can therefore refer to position $\im+1$ of $N$.

We choose an $s \in P_{\im+1}$ and start our word $u$ with the letter $\sigma_\im$ that satisfies
\begin{itemize}
\item $c_{\im+1} \in \sigma_\im(s,s')$ for all $s' \in S_\im$,
\item $c_i-1 \in \sigma_\im(s',s'')$ for all $i \leq \im$, $s' \in P_i'$, and $s'' \in S_i$, and,
\item if $c$ is odd, $c \in \sigma_\im(s',s'')$ for all $s',s'' \in  \reach(N)$.
\end{itemize}
No further priority is included in any set $\sigma_\im(s',s'')$ for $s',s'' \in Q$ and $s'' \in Q^\top$.

Then the effect on $N$ is obvious: the transition from $N$ reading $\sigma$ has co-priority $2\im+2$.
Starting from $N'$, the same state is reached.
The co-priority is $2\im+2$ if $N'$ has a position $(S_{\im+1}',c_{\im+1}',P_{\im+1}')$ with $s\in S_{\im+1}$ and $c_{\im+1}'=c_{\im+1}$, and $2\im+1$ otherwise.
(Note that, for the case that $s \in P_\im'$, this would imply $c_\im = c_{\im+1} + 2$.)

We can again continue $u$ in the same manner as above, although this results in the slightly shorter word $u = \sigma_{\im+1} \sigma_{\im+2}  \sigma_{\im+3}\ldots \sigma_m$.
The word $u$ will again satisfy the constraints; in particular $\rho(N',u,N)$ has purifying relevant change $i$ and $\rho(N,u,N)$ has co-priority $2\im+2$.
\end{proof}

\begin{lemma}
\label{lem:lose}
If any outgoing edge is removed from the verifier's centre vertex in any of these games, then the spoiler wins the language game.
\end{lemma}

\begin{proof}
If the spoiler has such an option, then \emph{she} can use $\mathcal D_n^c$ as a witness automaton.
Let $N$ be the spiked state, to whom the outgoing edge from the centre is removed.

Initially, the spoiler plays a word $u_0$ with $\rho(N_0,u_0,N)$ by choosing the edge $(v_0,u_0,c)$ from the initial vertex, such that $N$ is reached from the initial state of $\mathcal D_n^c$.
Henceforth she plays, whenever she is in a vertex $N'$, the word $u$ from the previous lemma by choosing the edge $(N',u,c)$.
This way, the two players construct a play
$(v_0,u_0,c) (c,\varepsilon,N_1) (N_1,u_1,c) (c,\varepsilon,N_2) (N_2,u_2,c) (c,\varepsilon,N_3)\linebreak(N_3,u_3,c) (c,\varepsilon,N_4) (N_4,u_4,c) \ldots$.

For every $i \geq 1$, the segment $\rho(N_i,u_i,N_{i+1})$ of the run of $\mathcal D_n^c$ on the word $w=u_0u_1u_2u_3u_4\ldots$ has even minimal co-priority.
Thus, the co-priority of the overall run is even.
\end{proof}

Together, the Lemmata \ref{lem:win}, \ref{lem:lose}, and \ref{lem:differentSets} provide:

\begin{theorem}
Every Rabin automaton $\mathcal R= \big(S,\Sigma_n^c,s_0,\delta, R\big)$ that recognises the complement language of $\mathcal P_n^c$ must, for each non-empty subset%
\footnote{For each $q{\in}Q$ there is only a single state $N$ of $\mathcal D_n^c$ with $\reach(N)=\{q\}$.}
$S\subseteq Q$ of the states $Q$ of $\mathcal P_n^c$, have at least as many states $s$ with $\reach(s)=S$ as $\mathcal D_n^c$ has spiked states $N$ with $\reach(N) = S$.
\end{theorem}

\begin{corollary}
Every Rabin automaton 
that recognises the complement language of $\mathcal P_n^c$ must contain at least as many states as $\mathcal D_n^c$ has spiked states.
\end{corollary}

By dualisation and the observation that parity automata are special Streett automata we simply get:

\begin{corollary}
A deterministic Streett or parity automaton that recognises the language of $\mathcal P_n^c$ must have at least as many states as $\mathcal D_n^c$ has spiked states.
\end{corollary}

The restriction to spiked states is minor -- using the estimation of 
Lemma \ref{lem:spiked}, we get:

\begin{theorem}
$\mathcal D_n^c$ has less than $1.5$ times as many states as the smallest deterministic Streett (or parity) automaton that recognises the language of $\mathcal P_n^c$.
\end{theorem}

State sizes for two parameters are usually not crisp to represent. But for the simple base cases, B\"uchi and one pair Rabin automata, we get very nice results: it establishes that the known upper bound for determinising B\"uchi to parity automata \cite{Schewe/09/determinise} are tight and Piterman's algorithm for it \cite{Piterman/07/Parity} is \textbf{optimal} modulo a factor of $3n$, where $2n$ stem from the fact that \cite{Piterman/07/Parity} uses state based acceptance.
With Lemma \ref{parityref} we get:

\begin{corollary}
The determinisation of B\"uchi automata to Streett or parity automata leads to $\theta(n!(n-1)!)$ states, and
the determinisation of one-pair Rabin automata to Streett or parity automata leads to $\theta(n!^2)$ states. 
\end{corollary}

\end{document}